\documentclass[journal]{IEEEtran}

\usepackage{amsmath, amssymb, comment,cite}
\usepackage{graphicx, color}
\usepackage[latin1]{inputenc}
\usepackage{subcaption}
\usepackage{floatrow}
\usepackage{relsize, bbm, mathtools}
\usepackage{algorithm,algorithmic}
\usepackage{balance}
\usepackage{float}
\floatstyle{plaintop}
\restylefloat{table}

\usepackage[nodisplayskipstretch]{setspace}
\newtheorem{lemma}{Lemma}
\newtheorem{remark}{Remark}
\newtheorem{theorem}{Theorem}
\newtheorem{corollary}{Corollary}

\newtheorem{proof}{Proof}

\newcounter{mytempeqcounter}

\newcommand{\bigformulatop}[2]{%
	\begin{figure*}[!t]
		\normalsize
		\setcounter{mytempeqcounter}{\value{equation}}
		\setcounter{equation}{#1}
		#2
		
		\setcounter{equation}{\value{mytempeqcounter}}
		\hrulefill
		\vspace*{4pt}
	\end{figure*}
}

\newcommand{\qa}{{\bf a}}

\newcommand{\qh}{{\bf h}}

\newcommand{\qr}{{\bf r}}
\newcommand{\qs}{{\bf s}}

\newcommand{\qv}{{\bf v}}
\newcommand{\qw}{{\bf w}}
\newcommand{\qx}{{\bf x}}
\newcommand{\qy}{{\bf y}}
\newcommand{\qz}{{\bf z}}

\newcommand{\qA}{{\bf A}}
\newcommand{\qB}{{\bf B}}
\newcommand{\qC}{{\bf C}}
\newcommand{\qD}{{\bf D}}

\newcommand{\qF}{{\bf F}}

\newcommand{\qH}{{\bf H}}
\newcommand{\qI}{{\bf I}}

\newcommand{\qQ}{{\bf Q}}

\newcommand{\qS}{{\bf S}}
\newcommand{\qT}{{\bf T}}

\newcommand{\qV}{{\bf V}}

\newcommand{\qX}{{\bf X}}

\newcommand{\dl}{{\mathsf{dl}}}
\newcommand{\SP}{{\mathsf{sp}}}
\newcommand{\EP}{{\mathsf{ep}}}
\newcommand{\Pil}{{\mathsf{che}}}
\newcommand{\TF}{{\mathsf{TF}}}
\newcommand{\ul}{{\mathsf{ul}}}
\newcommand{\Nguard}{{N_{\mathsf{guard}}}}

\newcommand{\tcp}{{T_\mathsf{CP}}}
\newcommand{\ofdm}{{\mathsf{ofdm}}}
\newcommand{\taudl}{\omega_{\dl}}
\newcommand{\tauul}{\omega_{\ul}^\Pil}
\newcommand{\tauulep}{\omega_{\ul}^\EP}
\newcommand{\tauulsp}{\omega_{\ul}^\SP}

\newcommand{\Sn}{\sigma_n^2}
\newcommand{\diag}{\mathrm{diag}}
\newcommand{\Trace}{\mathrm{Tr}}

\newcommand{\nm}{MN}

\newcommand{\chpqij}{{\chi}_{pq}^{ij}}
\newcommand{\kapqij}{{\kappa}_{pq}^{ij}}
\newcommand{\etpq}{{\eta}_{pq}}
\newcommand{\etpqr}{{\eta}_{pq^\prime}}
\newcommand{\betpqi}{\beta_{pq,i}}

\newcommand{\gampqi}{\gamma_{pq,i}^\Pil}
\newcommand{\gampqj}{\gamma_{pq,j}^\Pil}
\newcommand{\gampqri}{\gamma_{pq^\prime,i}^\Pil}
\newcommand{\gampqrj}{\gamma_{pq^\prime,j}^\Pil}
\newcommand{\Tpqi}{\qT_{pq}^{(i)}}
\newcommand{\Tpqj}{\qT_{pq}^{(j)}}
\newcommand{\Tpqjdg}{\qT_{pq}^{(j)^\dag}}
\newcommand{\Tpqidg}{\qT_{pq}^{(i)^\dag}}

\newcommand{\Hpq}{{\qH}_{pq}}
\newcommand{\hHpqdg}{\hat{\qH}_{pq}^\dag}
\newcommand{\hHpq}{\hat{\qH}_{pq}}

\newcommand{\hmki}{h_{pq,i}}
\newcommand{\hmkihat}{\hat{h}_{pq,i}^\Pil}
\newcommand{\hmkihatc}{\hat{h}_{pq,i}^{\Pil^*}}

\newcommand{\hmkjhatc}{\hat{h}_{pq,j}^{\Pil^*}}
\newcommand{\hmkjhat}{\hat{h}_{pq,j}^{\Pil}}

\newcommand{\epspqi}{\varepsilon_{pq,i}^{\Pil}}

\newcommand{\tmki}{\tau_{pq,i}}
\newcommand{\nmki}{\nu_{pq,i}}
\newcommand{\Lmk}{L_{pq}}

\newcommand{\Lmkp}{L_{pq'}}
\newcommand{\Bmki}{\beta_{pq,i}}
\newcommand{\Lmki}{\ell_{pq,i}}
\newcommand{\Lmkj}{\ell_{pq,j}}
\newcommand{\Kmki}{k_{pq,i}}
\newcommand{\Kmkj}{k_{pq,j}}
\newcommand{\kpmki}{\kappa_{pq,i}}
\newcommand{\kpmkj}{\kappa_{pq,j}}
\newcommand{\Kpq}{k_{pq}}
\newcommand{\Lpq}{\ell_{pq}}
\newcommand{\Pmax}{P_{\mathtt{max}}}

\newcommand{\Pplq}{P_{q}^{\mathtt{Pil}}}
\newcommand{\Pdtq}{P_{q}^{\mathtt{dt}}}
\newcommand{\rodtq}{\rho_{q}^{\mathtt{dt}}}
\newcommand{\rodtqr}{\rho_{q'}^{\mathtt{dt}}}
\newcommand{\roplq}{\rho_{q}^{\mathtt{Pil}}}
\newcommand{\Pplqr}{P_{q'}^{\mathtt{Pil}}}
\newcommand{\roplqr}{{\rho_{q'}^{\mathtt{Pil}}}}

\IEEEoverridecommandlockouts

\title{  Cell-Free Massive MIMO Meets OTFS Modulation }

\author{Mohammadali Mohammadi,~\IEEEmembership{Member,~IEEE,} Hien Quoc Ngo,~\IEEEmembership{Senior Member,~IEEE,}\\ and  Michail Matthaiou,~\IEEEmembership{Senior Member,~IEEE}\\
\thanks{Manuscript received Dec. 18, 2021; revised Apr. 11, 2022, July 26, 2022; accepted Sept. 26, 2022. Parts of this paper appeared at IEEE ICC 2022~\cite{MohammadICC2022}. The work of M. Mohammadi and M. Matthaiou was supported by a research grant from the Department for the Economy Northern Ireland under the US-Ireland R\&D Partnership Programme. The work of H. Q. Ngo was supported by the U.K. Research and Innovation Future Leaders Fellowships under Grant MR/S017666/1.  The associate editor coordinating the review of this paper and approving it for publication was A. Garc\'{\i}a Armada. (\emph{Corresponding author: Mohammadali Mohammadi.})

The authors are with the Centre for Wireless Innovation (CWI), Queen's University Belfast, U.K. Email:\{m.mohammadi, hien.ngo, m.matthaiou\}@qub.ac.uk.
}

}

\begin{document}
\maketitle
\vspace{-1em}
\begin{abstract}
We provide the first-ever performance evaluation of orthogonal time frequency space (OTFS) modulation in cell-free massive multiple-input multiple-output (MIMO) systems.  To investigate the trade-off between performance and overhead, we apply embedded pilot-aided and superimposed pilot-based channel estimation methods.  We then derive a closed-form expression for the individual user downlink and uplink spectral efficiencies (SEs) as a function of the numbers of APs, users and delay-Doppler domain channel estimate parameters. Based on these analytical results, we also present new scaling laws that the AP's and user's transmit power should satisfy, to sustain a desirable quality of service. It is found that when the number of APs, $M_a$, grows without bound, we can reduce the transmit power of each user and AP proportionally to $1/M_a$ and $1/M_a^2$, respectively, during the uplink and downlink phases. We compare the OTFS performance with that of orthogonal frequency division multiplexing (OFDM) at high-mobility conditions. Our findings reveal that, OTFS modulation with embedded pilot-based channel estimation provides up to $20$-fold gain over the OFDM counterpart in terms of $95\%$-likely per-user downlink SE. Finally, with superimposed pilot-based channel estimation, the increase in the uplink sum SE is more pronounced when the channel delay spread is increased.

\end{abstract}

\begin{IEEEkeywords}
	Cell-free massive multiple-input multiple-output (MIMO), orthogonal time frequency space (OTFS) modulation, spectral efficiency.
\end{IEEEkeywords}

\section{Introduction}~\label{Sec:Intro}
Future beyond-5G (B5G) wireless communication networks are envisioned to support the ever increasing demands for a wide range of network services as well as the seamless growth in the number of wireless connected devices~\cite{Tataria:6G}. In order to fulfill these demands, the cell-free massive multiple-input multiple-output (MIMO) concept has been recently proposed as a promising technique, thanks to its capability to support a high density of network devices, while providing substantial improvement in connectivity, spectral and energy efficiencies~\cite{Love:JSAC:2020,Hien:cellfree}.  In cell-free massive MIMO, there are no cell boundaries and a large number of access points (APs) are distributed over a large geographic area, and jointly serve many user equipments via time-division duplex (TDD) operation. The APs are connected via a front-haul network to the central processing units (CPUs), which are responsible for coordinating the coherent joint transmission and reception.  Although each AP is equipped with one or a small number of antennas, cell-free massive MIMO is still able to exploit both favorable propagation and channel hardening properties, similar to colocated massive MIMO~\cite{Hien:cellfree}. Therefore, simple beamforming schemes can be deployed at the APs to combine the received signals in uplink and to precode the symbols in the downlink~\cite{Zhang:TCOM:2021,Matthaiou:TCOM:2021}. Moreover, with channel hardening, decoding the signals using the channel statistics can provide sufficiently good performance~\cite{Hien:cellfree,Hien:TGCN:2018}. Nevertheless, high-mobility scenarios, such as high-speed railways~\cite{railway},  vehicle to vehicle (V2V) communications, and unmanned aerial vehicles (UAV) communications have remained largely unexplored. These systems, wherein a large number of high-mobility users require high data rates services with low latencies, will challenge the existing designs of cell-free massive MIMO due to the high-mobility channel interfaces.

In high-mobility scenarios, accommodating various type of emerging services, such as virtual reality, augmented reality, high quality video conferencing, and online gaming, fundamental limitations must be overcome pertaining to rapid channel variations. Recall that wireless channels in such scenarios are inherently linear time-variant fading channels, also known as \emph{doubly-selective} or \emph{doubly-dispersive} channels~\cite{Jakes}. Therefore, the widely adopted orthogonal frequency division multiplexing (OFDM) modulation in 4G and 5G becomes a problematic option for high-mobility scenarios. Although OFDM can efficiently mitigate inter-symbol interference (ISI) induced by  time dispersion through the use of cyclic prefix (CP), the frequency dispersion caused by  Doppler shifts impairs the orthogonality across the OFDM sub-carriers. The imperfect orthogonality of the sub-carriers is detrimental for the OFDM receiver, since it renders  OFDM incapable of yielding a robust performance.  Although the OFDM symbol duration can be shortened, so that the channel variations over each symbol appear inconsequential, one major drawback is the reduced spectral efficiency due to the CP, which undermines the application of OFDM in high data rate application scenarios~\cite{Jakes}. In this regard, designing new modulation waveforms  for high-mobility communications is indispensable~\cite{Wei:WC:2021}.

Recently, Hadani \textit{et al} developed a new two-dimensional (2D) modulation, referred to as orthogonal time-frequency space (OTFS) modulation, which has shown significant efficiency in dealing with the high Doppler problems occurring in OFDM modulation~\cite{Hadani:WCNC:2017}. By leveraging the 2D inverse symplectic finite Fourier transform (ISFFT),  OTFS multiplexes the information symbols in the delay-Doppler (DD) domain rather than in the time-frequency (TF) domain as in OFDM modulation. More specifically, through the 2D transformation from the DD domain to the TF domain, each information symbol will span the entire TF domain channel over an OTFS frame. Therefore,  OTFS efficiently exploits the potential of full-diversity, which is the key for supporting ultra-reliable communications~\cite{Wei:WC:2021}. More importantly, OTFS incorporates the key principle of the DD domain representation, where the number of channel taps correspond to the group of reflectors having a particular delay and Doppler value. Since the relative velocity and distance of the reflectors with respect to the receiver can be considered roughly invariant in the duration
of a frame, the DD taps are time invariant for a larger observation time as in the case of the TF representation~\cite{Hadani:WCNC:2017}.
Moreover, since there are only a small number of group of reflectors with different delay and Doppler values, the parameters that need to be estimated are also fewer, and the representation in this domain is more compact and sparse~\cite{Raviteja:TWC:2018,Raviteja:TVT:2019}. Finally, from an implementation point of view, OTFS modulation can be realized by adding pre-processing and post-processing blocks to any conventional multi-carrier modulations, such as OFDM modulation. Recognizing the aforementioned benefits of OTFS, there has been an upsurge of interest in applying OTFS to different wireless communication systems, including millimeter wave (mmWave) and terahertz communication systems~\cite{Chavva:WCNC:2021,Kim:ICC:2021}, non-orthogonal multiple access communication networks~\cite{Ding:TCOM:2021,Sharma:TCOM:2021}, and  massive MIMO~\cite{Wang:JSAC:2020,Dobre:JSAC:2021,Shi:TWC:2021,Wang:TVT:2021,Raviteja:TVT:2021}.

We note that the existing research in the context of massive MIMO has overlooked the importance of a spectral efficiency (SE) analysis and has merely focused on the pilot design and channel estimation for colocated massive MIMO-OTFS systems. Only one recent work~\cite{Raviteja:TVT:2021}, has proposed a low complexity multi-user precoding and detection methods for colocated massive MIMO OTFS systems and derived an expression for the achievable downlink SE of the system. The authors in~\cite{Wang:JSAC:2020} proposed an uplink-aided high-mobility downlink channel estimation scheme for massive MIMO-OTFS networks, which recover the uplink channel parameters including the angles, delays, Doppler frequencies, and the channel gains. In~\cite{Shi:TWC:2021}, a downlink  channel state information (CSI) acquisition scheme was proposed in the presence of fractional Doppler, including a deterministic pilot design and channel estimation algorithm. The authors in~\cite{Dobre:JSAC:2021} proposed a path division multiple access scheme for both uplink and downlink, where a 3D Newtonized orthogonal matching pursuit algorithm is utilized to recover the uplink channel parameters, including channel gains, direction of arrival, delays, and Doppler frequencies, over the antenna-time-frequency domain. In~\cite{Wang:TVT:2021}, a low-overhead and low-complexity receiver scheme was developed, including the pilot pattern design,  channel estimation, symbol detection, and carrier frequency offset compensation.

Nevertheless, the insights for colocated massive MIMO design cannot be directly extrapolated to the cell-free massive MIMO setting, as the latter is fundamentally different due to (i) the distributed nature of both APs and users, (ii) the large number of users and APs which render the multiple access schemes impractical, and (iii) the use of downlink transmission power control at the APs for each user. Furthermore, it is still not clear whether in asymptotic scenarios (i.e., when the number of APs is high) the effects of non-coherent interference, small-scale fading, and noise disappear in cell-free massive MIMO systems with OTFS modulation or not?

In order to serve  high-mobility users through time-variant channels, the integration of OTFS modulation into cell-free massive MIMO emerges as a promising alternative. To the authors' best knowledge, the consolidation of OTFS modulation with cell-free massive MIMO has not been reported before. Thus, this paper will focus on the achievable SE analysis of OTFS in cell-free massive MIMO. The main contributions of our work are as follows:

\begin{itemize}
\item We apply both embedded pilot (EP)-aided channel estimation with reduced guard interval and superimposed pilot (SP)-based channel estimation on OTFS modulation-based cell-free massive MIMO systems and derive the  minimum mean-square error (MMSE) estimate of the channel gains at the APs.
\item We investigate the downlink and uplink SE of OTFS waveforms in cell-free massive MIMO with maximum-ratio processing.\footnote{Maximum-ratio (MR) processing (a.k.a. matched	filtering or conjugate beamforming), is computationally simple and can be implemented in a distributed manner, i.e.,  most processing is done locally at the APs. Although other linear processing techniques, such as  minimum mean-square error/zero-forcing, may improve the system performance, they require more backhaul signalling for CSI exchange. Moreover, their implementation requires the computation of the inverse of a matrix that scales with the number of APs and antennas and, thus, induces a prohibitively intensive computation cost. This complexity is dramatically increased in OTFS systems, since the effective DD channel, which is used for precoding/combing design, has a very high dimension.} We derive closed-form expression for the individual user downlink and uplink achievable SE for finite numbers of APs and users, taking into account the effects of channel estimation errors.
\item We elaborate on the power-scaling laws as follows: If CSI is estimated with uncertainty, then when the number of APs, i.e., $M_a$, gets asymptotically large, cell-free massive MIMO will still bring significant power savings for each user and AP. In particular, when MR processing is applied, we can reduce the transmit power of each user proportionally to $1/M_a$ during the uplink transmissions. Moreover, the transmit power of each AP can be scaled down by  $1/M_a^2$ during downlink transmission to obtain the same rate.

\item We present the achievable uplink SE expressions for the minimum mean-squared error-based successive interference cancellation (MMSE-SIC) detector and arbitrary combining schemes with centralized and distributed processing designs. This analytical framework provides a basis for numerically evaluating the benefits and costs (in terms of backhaul signaling and computational complexity) of the different implementations.

\item Our findings demonstrate that a significant performance improvement can be achieved  by the OTFS  over the OFDM modulation in cell-free massive MIMO over high-mobility channels. Moreover, local minimum mean square error (L-MMSE) processing with simple centralized decoding significantly improves the uplink SE compared with MR combining at the expense of high computational complexity.
\end{itemize}

The rest of the paper is organized as follows. In Section~\ref{Sec:SysModel},
the system model and channel model are introduced. Two channel estimation schemes are presented in Section~\ref{Sec:CHE}. In Section~\ref{Sec:Perf}, we provide the achievable downlink and uplink SE and also investigate the power-scaling laws. Numerical results are presented in Section~\ref{Sec:Numer}, followed by conclusions in Section~\ref{Sec:conclusion}.

\textit{Notation:} We use bold upper case letters to denote matrices, and bold lower case letters to denote vectors; The superscripts $(\cdot)^T$, $(\cdot)^*$, and $(\cdot)^\dag$ stand for the transpose, conjugate and conjugate-transpose, respectively; $|\cdot|$ denotes the absolute value of a complex scalar; $\det(\cdot)$, $\Trace(\cdot)$, $(\cdot)^{-1}$, and $\mathrm{vec}(\cdot)$ denote the determinant,  trace, the inverse, and the vectorization operation;  $\mathbb{C}^{L\times N}$ denotes a $L\times N$ matrix; $\diag\{\cdot\}$ returns a diagonal matrix; $\mathrm{circ}\{\qx\}$ represents a circulant matrix whose first column is $\qx$; The matrix $\qF_N=\Big(\frac{1}{\sqrt{N}} e^{-j2\pi\frac{k\ell}{N}}\Big)_{k,\ell=0,\ldots,N-1}$ denotes the unitary discrete Fourier transform (DFT) matrix of dimension $N \times N$; $\qI_M$ and $\boldsymbol{0}_{M\times N}$ represent the $M\times M$ identity matrix and zero matrix  of size $M\times N$, respectively; $[\qA]_{(i,j)}$, $[\qA]_{(i,:)}$, and $[\qA]_{(:,j)}$ denote the  $(i,j)$th entry, the $i$th row, and $j$th column of $\qA$, respectively. The operator $\otimes$ denotes the Kronecker product of two matrices; $(\cdot)_N$ denotes the modulo $N$ operation; $\mathfrak{R}[\cdot]$ returns the real part of the input complex number; $\mathbb{N}[a,b]$ represents the set of non-negative integer numbers ranging from $a$ to $b$; $\lfloor\cdot\rfloor$ is the floor function which returns the largest integer smaller than the input value. The big-O notation $\mathcal{O}(\cdot)$ describes asymptotically the order of computational complexity. Finally, $\mathbb{E}\{\cdot\}$ denotes the statistical expectation.

\section{System Model}~\label{Sec:SysModel}
We consider a cell-free massive MIMO system consisting of $M_a$ APs and $K_u$ users.  The APs and users are all equipped with a single antenna, and they are randomly located in a large area. The APs are connected to a CPU via a front-haul network. The users are assumed to move at high speeds, thus the channels between the AP and users experience doubly-selective fading. An OTFS frame is divided into two phases: uplink payload transmission with channel estimation, and downlink payload transmission. Uplink and downlink transmissions are underpinned by TDD operation.

\subsection{OTFS Modulation and Channel Model}~\label{Sec:OTFS}
Consider an OTFS system with $M$ sub-carriers having $\Delta f$ (Hz) bandwidth each, and $N_T=N_{\dl}+N_{\ul}$ symbols having $T$ (seconds) symbol duration, of which $N_{\ul}$ symbols are dedicated for uplink data transmission as well as channel estimation and $N_{\dl}$ symbols are used for downlink data transmission. For the sake of clarity and brevity, we assume $N_{\mathtt{dl}}=N_{\mathtt{ul}}=N$. Nevertheless, in  the final results, $N$ can be readily replaced by $N_{\mathtt{dl}}$  and $N_{\mathtt{ul}}$, in the corresponding variables for downlink and uplink results, respectively. Therefore, the total bandwidth of the system is $M\Delta f$ and $NT$ is the duration of an OTFS block during downlink or uplink transmission.

The modulated data symbols of the $q$th user $\{x_q[k,\ell] , k\in \mathbb{N}[0,N-1], \ell\in \mathbb{N}[0,M-1]\}$ are arranged over the DD lattice $\Lambda=\left\{\frac{k}{NT}, \frac{\ell}{M\Delta f}\right\}$, where $\frac{1}{NT}$  and $\frac{1}{M\Delta f}$ are the sampling intervals of the Doppler-dimension and  delay-dimension, respectively; $k$ and $\ell$ represent the Doppler shift and delay index, respectively. An ISFFT is applied at the OTFS transmitter to convert the set of $\nm$ zero-mean independent and identically distributed (i.i.d.) data symbols $x_q[k,\ell]$, with $\mathbb{E}\{|x_q[k,\ell]|^2\}=\Pdtq$ being the $q$th user transmit power in the uplink, to the symbols $X_q[n,m]$ in the TF domain as
\begin{align}~\label{eq:Xtf}
X_q[n,m] = \frac{1}{\sqrt{MN}}\sum_{k=0}^{N-1}\sum_{\ell=0}^{M-1} x_q[k,\ell] e^{j2\pi\left(\frac{nk}{N}-\frac{m\ell}{M}\right)},
\end{align}
where $n\in\mathbb{N}[0,N-1]$, and $m\in\mathbb{N}[0,M-1]$. Accordingly, by using a Heisenberg transform, $X_q[n,m]$ are converted to a time domain signal as
\begin{align}~\label{eq:st}
s_q(t)=\sqrt{{\eta}_q}\sum_{n=0}^{N-1}\sum_{m=0}^{M-1}
X_q[n,m] g_{tx}(t-nT)
e^{j2\pi m \Delta f(t-nT)},
\end{align}
where $0\leq{\eta}_q\leq1$ is the uplink power control coefficient for the $q$th user and $g_{tx}(t)$ is the transmitter pulse of duration $T$  which is defined as
\begin{align}~\label{eq:gtx}
g_{tx}(t) =\left\{ \begin{array}{ll}
\!\!\!\frac{1}{\sqrt{T}}         & 0\leq t\leq T,  \\
\!\!\!0           &  \text{otherwise}.\end{array} \right.
\end{align}

The channel impulse response between the $q$th user and the $p$th AP in the DD domain is given by~\cite{Raviteja:TWC:2018}
\begin{align}~\label{eq:hmk}
h_{pq}(\tau,\nu) = \sum_{i=1}^{\Lmk}\hmki\delta(\tau-\tmki)\delta(\nu-\nmki),
\end{align}
where $\Lmk$ denotes the number of paths from the $q$th user to the $p$th AP, $\tmki$, $\nmki$, and  $\hmki$ denote the delay, Doppler shift, and the channel gain, respectively, of the $i$th path of the $q$th user to the $p$th AP. The complex channel gains $\hmki$  for different $(pq, i)$ are independent random variables (RVs) with $\hmki\sim\mathcal{CN}(0,\Bmki)$. The delay and Doppler shifts for the $i$th path are given by $\tmki = \frac{\Lmki}{M\Delta f}$ and $\nmki = \frac{\Kmki+\kpmki}{N T}$
respectively, where $\Lmki \in \mathbb{N}[0, M-1]$ and $\Kmki\in \mathbb{N}[0, N-1]$ are the delay index and Doppler index of the $i$th path, and $\kpmki\in (-0.5,0.5)$ is a fractional Doppler associated with the $i$th path.  Let $\Kpq$ and $\Lpq$ denote the delay and Doppler taps corresponding to the largest delay and Doppler between the $q$th user and the $p$th AP. We note that the typical value of the sampling time $1/(M\Delta f)$ is usually sufficiently small in the delay domain. Hence, the impact of fractional delays in typical wideband systems can be neglected~\cite{Raviteja:TWC:2018}.\footnote{In practice, with $\Delta f=15$ kHz and $M=512$ ($1024$), a typical value of the sampling time $1/(M\Delta f)$ in the delay domain is $0.13~\mu$s ($0.065~\mu$s). Such sampling times are sufficiently small, allowing different paths, with delay difference of up to $0.13~\mu$s ($0.065~\mu$s) to be efficiently separable at the receiver.}

\subsection{Uplink Payload Data Transmission}~\label{Sec:ULdata}
In the uplink, all $K_u$ users simultaneously send their data to the APs. The received signal at the $p$th AP is expressed as
\begin{align}~\label{eq:rt}
r_p(t)\!=\!
\!\sum_{q=1}^{K_u}\!\!
\sqrt{{\eta}_q}
\int\!\!\!\int \!\!\!
h_{pq}(\tau,\nu)s_q(t-\tau)e^{j2\pi\nu(t-\tau)} d\tau d\nu \!+\! w_p(t),
\end{align}
where $w_p(t)$ represents the noise signal in the time domain following a stationary Gaussian random process and we have $w_p(t)\sim\mathcal{CN}(0,\Sn)$ with $\Sn$ denoting the noise variance. The received signal is processed via a Wigner transform, implemented by a receiver filter with an impulse response $g_{rx}(t)$ (the receive filter has the same definition as $g_{tx}(t)$ in~\eqref{eq:gtx}) followed by a sampler, to obtain the received samples $\{Y_p[n,m], n\in\mathbb{N}[0,N-1], m\in\mathbb{N}[0,M-1]$\} in the TF domain
\begin{align}~\label{eq:Ymn}
Y_p[n,m]=\int r_p(t) g_{rx}(t-nT)
e^{-j2\pi m \Delta f (t-nT)} dt.
\end{align}

Finally, by applying a SFFT to $Y_p[n,m]$, assuming that  practical non-ideal rectangular waveforms are used for transmit and receive pulse shaping filters~\cite{Raviteja:TWC:2018}, the received signal at the $p$th AP in the DD domain can be written as
\begin{align}~\label{eq:yAP:rec:frac}
y_p[k,\ell]
\!&=
\sqrt{{\eta}_q}
\sum_{q=1}^{K_u}\!
\sum_{k'=0}^{\Kpq}
\sum_{\ell'=0}^{\Lpq}
b[k',\ell']
\!\!
\sum_{c=-N/2}^{N/2-1}\!
\!
{h}_{pq}[k',\ell']\alpha[k,l,c]\nonumber\\
&\hspace{2em}
\times x_q[(k\!-\!k'\!+\!c)_N,(\ell\!-\!\ell')_M]
\!+\! w_p[k,\ell],
\end{align}
where $b[k',\ell']\in\{0,1\}$ is a path indicator, i.e., $b[k',\ell']=1$ indicates there is a path with Doppler tap $k'$ and delay tap $\ell'$, otherwise $b[k',\ell']=0$ (i.e., $\sum_{k'=0}^{\Kpq}\sum_{\ell'=0}^{\Lpq} b[k',\ell']=\Lmk$). Moreover $\alpha[k,l,c]$ is given by~\cite{Raviteja:TWC:2018}
\begin{align}
\alpha[k,\ell,c] \!= \!
	\left\{
	\begin{array}{ll}
			\!\!\!\frac{1}{N}\beta_i(c) e^{-j2\pi\frac{(\ell-\ell')(k'+\kappa')}{MN}}\\
			    								 & \hspace{-5em}    \ell'\leq \ell < M  \\
			\!\!\!\frac{1}{N}(\beta_i(c)\!-\!1) e^{-j2\pi\frac{(\ell\!-\ell')(k'+\kappa')}{MN}}e^{-j2\pi\frac{(k-k'+c)_N}{N}}\\
													& \hspace{-5em}  0\leq \ell <\ell', \end{array} \right.
\end{align}
where $\beta_i(c) \triangleq\ \frac{ e^{-j2\pi(-c-\kappa')} -1}{e^{-j\frac{2\pi}{N}(-c-\kappa')} -1}$, whilst $\kappa'$ denotes the fractional Doppler associated with the $(k',\ell')$ path. Moreover, in~\eqref{eq:yAP:rec:frac}, $ w_p[k,\ell]$ is the received additive noise, which by using~\eqref{eq:rt} and~\eqref{eq:Ymn} can be expressed as
\begin{align*}
w_p[k,\ell] = \frac{1}{\sqrt{MN}}\sum_{n=0}^{N-1}\sum_{m=0}^{M-1} W_p[n,m] e^{-j2\pi(\frac{nk}{N}-\frac{m\ell}{M})},
\end{align*}
where $W_p[n,m]$ is the received noise sampled at $t=nT$ and $\nu=m\Delta f$, given by $W_p[n,m]=\int w_p(t) g_{rx}(t-nT) e^{-j2\pi m \Delta f (t-nT)} dt$. It can be readily checked that since $w_p(t)\sim\mathcal{CN}(0,\Sn)$, we have also  $w_p[k,\ell]\sim\mathcal{CN}(0,\Sn)$.

For the sake of simplicity of presentation and analysis, we consider the vector form representation of the input-output relationship of OTFS system in the DD domain. Let $\qv\in\{\qx_q, \qy_p, \qw_p\}\in \mathbb{C}^{\nm\times 1}$ where $\qx_q$, $\qy_p$, and $\qw_p$ denote the vector of transmitted symbols from the $q$th user, received signal vector at the $p$th AP, and the corresponding noise vector, respectively (e.g., $\qx_q=\mathrm{vec}(\qX_q)$ where $\qX_q\in\mathbb{C}^{M\times N}$).
Hence, the input-output relationship in~\eqref{eq:yAP:rec:frac} can be expressed in vector form as
\begin{align}~\label{eq:yAPm:Vect}
\qy_p = \sum_{q=1}^{K_u} \sqrt{\rodtq \eta_q}\qH_{pq} \tilde{\qx}_q + \qw_p
\end{align}
where  $\tilde{\qx}_q=\frac{1}{\sqrt{\Pdtq}}\qx_q\in\mathbb{C}^{MN\times 1}$;   $\rodtq=\frac{\Pdtq}{\Sn}$ is the normalized uplink signal-to-noise ratio (SNR); $\qH_{pq}\in \mathbb{C}^{MN\times MN}$ is the effective DD domain channel between the $q$th user and $p$th AP, given by~\cite{KWAN:TWC:2021}
\begin{align}~\label{eq:Hpq}
\qH_{pq}
&=
\sum_{i=1}^{\Lmk}
\hmki \qT_{pq}^{(i)},
\end{align}
where $\qT_{pq}^{(i)}=({\qF}_N \otimes {\qI}_M)\boldsymbol{\Pi}^{\Lmki} \boldsymbol{\Delta }^{\Kmki+\kpmki}
({\qF}_N^\dag \otimes {\qI}_M)$, while $\boldsymbol{\Pi} =\mathrm{circ}\{[0,1,0,\ldots,0]^T_{MN\times 1}\}$ denotes a $\nm\times \nm$ permutation matrix and $\boldsymbol{\Delta} = \diag\{z^0,z^1,\ldots,z^{MN-1}\}$ is a diagonal matrix with $z=e^{\frac{j2\pi}{\nm}}$.

In order to detect the symbol transmitted from the $q$th user, $\tilde{\qx}_q$, the $p$th AP multiplies the received signal $\qy_p$ with the Hermitian of the (locally obtained) channel estimation matrix $\hat{\qH}_{pq}$. In Section~\ref{Sec:CHE}, we discuss the channel estimation of cell-free massive MIMO systems.
\subsection{Downlink Payload Data Transmission}~\label{Sec:DLData}
The APs use maximum-ratio beamforming to transmit signals to $K_u$ users. Let $\qs_q=\mathrm{vec}(\qS_q)\in\mathbb{C}^{\nm\times 1}$ be the intended signal vector for the $q$th user, where $\qS_q\in\mathbb{C}^{M\times N}$ representing the symbols in the DD domain, whose $(k,\ell)$th element $s_q[k,\ell]$ is the modulated signal in the $k$th Doppler and $\ell$th delay grid, for $ k\in \mathbb{N}[0,N-1], \ell\in \mathbb{N}[0,M-1]$. Therefore, the signal transmitted   from the $p$th AP is
\begin{align}~\label{eq:xqd}
\qx_{\dl,p}= \sqrt{\rho_d}
\sum_{q=1}^{K_u}
\eta_{pq}^{1/2}\hat{\qH}_{pq}^\dag \qs_q,
\end{align}
where $\rho_d$ is the normalized SNR of each symbol;  $\eta_{pq}$, $p=1,\ldots,M_a$, $q=1,\ldots,K_u$ are the power control coefficients chosen to satisfy  the  following power constraint at each AP~\cite{Hien:cellfree}
\begin{align}~\label{eq:AP:powconst}
\mathbb{E}\left\{\|\qx_{\dl,p}\|^2\right\} \leq \rho_d.
\end{align}

The received signal at the $q$th user in DD domain can be expressed as
\vspace{-0.2em}
\begin{align}~\label{eq:zqd}
\qz_{\dl,q}
&= \sum_{p=1}^{M_a}{\qH}_{pq} \qx_{\dl,p} + \qw_{\dl,q}\nonumber\\
&=\sqrt{\rho_d}\sum_{p=1}^{M_a}
\eta_{pq}^{1/2} {\qH}_{pq} \hat{\qH}_{pq}^\dag \qs_{q}\nonumber\\
&\hspace{1em}
+
\sqrt{\rho_d}\sum_{p=1}^{M_a}
\sum_{q'\neq q}^{K_u}
\eta_{pq'}^{1/2}{\qH}_{pq}\hat{\qH}_{pq'}^\dag \qs_{q'} +\qw_{\dl,q},
\end{align}
where $\qw_{\dl,q}\in\mathbb{C}^{\nm\times 1}$ is the AWGN vector at the user $q$.

\section{Channel Estimation}~\label{Sec:CHE}
By considering TDD operation, we rely on channel reciprocity to acquire CSI in the cell-free massive MIMO system with OTFS modulation.\footnote{Since the relative velocity and distance of the reflectors remain roughly the same for at least few milliseconds, the DD taps are time invariant for a larger observation time as compared to that in the TF representation~\cite{Hadani:2016}. Therefore, through the proper design of OTFS parameters, it is reasonable to assume the same DD impulse response for both the uplink and downlink transmission phases.} An intuitive method to estimate CSI is to transmit an impulse in the DD domain as the training pilot and then estimate the DD channel impulse response using the least square (LS) or MMSE estimator~\cite{Raviteja:TVT:2019,KWAN:TWC:2021,Flanagan:vtc:2020}. The transmitted impulse pilot is spread by the channel and interfere with data symbols in the DD domain. Therefore, inserting guard symbols to avoid the interference between the pilot and data symbols is required. This method is termed as EP-aided channel estimation in the literature~\cite{Raviteja:TVT:2019}. We further notice that a few recent works have investigated the pilot pattern design and channel estimation for MIMO-OTFS~\cite{Hanzo:TCOM:2021} and colocated massive MIMO~\cite{Shi:TWC:2021,Heath:SPL:2019,Fei:WCL:2021}. However, these pilot patterns cannot be directly applied to cell-free massive MIMO systems due to their high computational complexity and huge pilot overhead, which scales with the number of users. In order to illustrate the required pilot overhead in more detail, let us consider the EP-aided scheme. In the considered system, to distinguish the DD channels associated with $K_u$ users, $K_u$ impulses are required to be transmitted. Let $\tau_{max}$ and $\nu_{max}$ be the maximum delay and the maximum Doppler spread among all channel paths. Define $\ell_{max} = \tau_{max}M\Delta f=\max_{p,q}\ell_{pq}$ and $k_{max} = \nu_{max}NT=\max_{p,q}k_{pq}$, which indicates that the DD channel responses of the users have a finite support $[0, \ell_{max}]$ along the delay dimension and $[-k_{max}, k_{max}]$ along the Doppler dimension. Then, the guard intervals between two adjacent impulse along the Doppler and delay dimension should not be smaller than $2k_{max}$ and $\ell_{max}$, respectively. Moreover, to avoid inter-user interference during the channel estimation, users cannot use dedicated pilot and guard grids of each other for data transmission.  As a result, the pilot length to transmit $K_u$ impulses in OTFS-based cell-free massive MIMO systems should be at least $2K_uk_{max}\ell_{max}$. This would be more challenging in the case of fractional Doppler, where by using full-guard pilot pattern~\cite{Raviteja:TVT:2019,KWAN:TWC:2021} the length of pilot overhead should be $2K_u N \ell_{max}$.

In order to apply the EP-aided channel estimation method with reduced guard symbols into cell-free massive MIMO, we relax the assumption that users cannot transmit their information symbols over the dedicated channel estimation DD grids of each other. Specifically, we assume that the pilot and guard region of different users cannot overlap, however,  user $q'$ can use the pilot and guard DD grid of the other users, $q=1,\ldots, K_u$,  $q\neq q'$ for the uplink data transmission. Therefore, the pilot overhead for user $q'$ is $\Nguard$ symbols and  there are still $(M N_{\ul}-\Nguard)$ DD resource bins available for uplink data transmission. Assuming that the pilot and guard DD grid of different users are regularly placed next to each other, the total overhead is $\Nguard=(2\ell_{max}+1)(4k_{max}+4\hat{k}+1)$ per each user, where $\hat{k}$ denotes the additional guard to mitigate the spread due to fractional Doppler and $\hat{k}\in\big\{0,\ldots,\lfloor \frac{N-4k_{max}-1}{4}\rfloor\big\}$. Therefore,  $K_u\leq \left \lfloor \frac{\nm}{\Nguard}\right \rfloor$ users can be supported. As an alternative, pilot symbols can be superimposed on to the data symbols in the DD domain to avoid the SE loss due to the null guard interval transmission~\cite{Dataided:WCL:2021,Prasad:TWC:2021}. This scheme is termed as SP-based channel estimation, wherein a complete transmit frame size of $M\times N$ contains superimposed data and pilot symbols~\cite{Prasad:TWC:2021}. This provides a degree-of-freedom for orthogonal pilot assignment in the considered system, where $\nm$ orthogonal pilot sequences can be designed. Therefore, SP-based channel estimation can serve $K_u\leq \nm$ users at the same time without pilot contamination.

\subsection{Embedded-Pilot Channel Estimation}~\label{sec:che}
We deploy the EP-aided channel estimation method with reduced guard symbols, while users are allowed to use the dedicated pilot and guard DD grids of each other's for data transmission. Consider $\varphi_q[k_q, \ell_q]$, with $\mathbb{E}\{|\varphi_q[k_q, \ell_q]|^2\}=\Pplq$ denoting a known pilot symbol for the $q$th user at a specific DD grid location $[k_q, \ell_q]$,  $x_{dq}[k,\ell]$, with $\mathbb{E}\{|x_{dq}[k,\ell]|^2\}=\Pdtq$ denoting the $q$th user's data symbol at grid point $[k,\ell]$, and assume that each pilot is surrounded by a guard region of zero symbols.  Therefore, for the $q$th user the pilot, guard, and data symbols in the DD grid are arranged as
\begin{align}~\label{eq:datapatern}
x_q[k,\ell] = \left\{ \begin{array}{ll}
\!\!\!\varphi_q         &  k=k_q, \ell=\ell_q,  \\
\!\!\!0           &  k\in \mathcal{K}, k\neq k_q ~  \ell\in\mathcal{L}, \ell\neq\ell_q,\\
\!\!\!x_{dq}[k,\ell] &  \mbox{otherwise},\end{array} \right.
\end{align}
where $\mathcal{K}=\{k_q-2k_{max}-2\hat{k},\cdots, k_q+2k_{max}+2\hat{k}\}$, $\mathcal{L}=\{\ell_q-\ell_{max}\leq \ell \leq \ell_q+\ell_{max}\}$.  At the receiver,  the received symbols $y_p[k,\ell]$, $k_q-k_{max}-\hat{k} <k< k_q+k_{max}+\hat{k}$, $\ell_q\leq \ell \leq \ell_q-\ell_{max}$ are used for channel estimation. Therefore, from~\eqref{eq:yAP:rec:frac}, we have
\begin{align}~\label{eq:yAPpil}
y_p[k,\ell] &=
\sqrt{\eta_q}
\tilde{b}[\ell-\ell_q] \tilde{h}_{pq}[(k-k_q)_N,(\ell-\ell_q)_M]\varphi_q\nonumber\\
&\hspace{2em}+\mathcal{I}_1(k,\ell)+\mathcal{I}_2(k,\ell)+ w_p[k,\ell],
\end{align}
where
\begin{align*}
\tilde{b}[\ell-\ell_q] =
\left\{ \begin{array}{ll}
1,        &  \sum_{k'=0}^{k_{pq}} b[k',\ell-\ell_q]\geq 1, \\
0           &  \mbox{otherwise},\end{array} \right.
\end{align*}
is the path indicator and
\begin{align*}
\tilde{h}_{pq}[(k-k_q)_N,(\ell\!-\!\ell_q)_M] &\nonumber\\
&\hspace{-3em}
=\sum_{k'=0}^{k_{pq}}
b[k',\ell\!-\!\ell_q] h_{pq}[k',\ell\!-\!\ell_q]\alpha(k,\ell,c).
\end{align*}

In~\eqref{eq:yAPpil}, $\mathcal{I}_1(k,\ell)$ denotes the interference spread from $q$th user's data symbols due to the existence of fractional Doppler, given by
\vspace{-0.1em}
\begin{align}~\label{eq:Ikl1}
\mathcal{I}_1(k,\ell) &\!=\!
\sqrt{\eta_q}\sum_{k'=0}^{k_{pq}}
\sum_{\ell'=0}^{\ell_{pq}}
b[k',\ell']
\sum_{c\not\in \mathcal{K}}
\!
h_{pq}[(k\!-\!k')_N,(\ell\!-\!\ell')_M]\nonumber\\
&\hspace{1em}\times
\alpha[k,\ell,c]x_{dq}[(k\!-\!k'\!+\!c)_N,(\ell\!-\!\ell')_M],
\end{align}
and  $\mathcal{I}_2(k,\ell)$ denotes the interference spread from data symbols of other users, which can be expressed as
\vspace{0em}
\begin{align}~\label{eq:Ikl2}
\mathcal{I}_2(k,\ell) &= \sum_{q'\neq q}^{K_u}\sqrt{\eta_{q^\prime}}
\sum_{k'=0}^{k_{pq'}}
\sum_{\ell'=0}^{\ell_{pq'}}
b[k',\ell']
\sum_{c=-N/2}^{N/2}
h_{pq'}[k',\ell']\nonumber\\
&\hspace{0em}
\times
\alpha(k,\ell,c)x_{dq'}[(k-k'+c)_N,(\ell-\ell')_M].
\end{align}

Note that the parameters $\ell_{pq,i}$ and $k_{pq,i}$ remain constant over multiple OTFS frames, while the complex channel gain $h_{pq,i}$  varies across frames. Therefore, we can efficiently estimate these parameters using the proposed approaches in ~\cite{Raviteja:TVT:2019,Flanagan:vtc:2020} and use these results during the MMSE	estimation process of the channel gains. Then, by using the MMSE estimation,  $h_{pq,i}$ can be estimated as $\hat{h}_{pq,i}^{\EP} = c_{pq,i}y_p[k,\ell],$ where\footnote{Hereafter, we use the superscripts $\EP$ and $\SP$ in the related variables to denote the EP and SP-based channel estimation, respectively. }
\begin{align}~\label{eq:cpqk}
c_{pq,i}\!=\!
\frac{\sqrt{\Pplq\eta_q}\beta_{pq,i}}
{\Pplq\eta_q\beta_{pq,i}
	\!+\! \mathbb{E}
	\{|\mathcal{I}_1(k,\ell)|^2\} \!+\!
	\mathbb{E}
	\{|\mathcal{I}_2(k,\ell)|^2\}\!+\!
	\Sn}.
\end{align}
We now proceed to derive the two expectation terms in the denominator of~\eqref{eq:cpqk}.
By invoking~\eqref{eq:Ikl1}, and considering the fact that data symbols are independent and $\mathbb{E}\{|x_{dq}[k',\ell']|^2\}=P_u$, we obtain $\mathbb{E}\{|\mathcal{I}_1(k,\ell)|^2\}$ as~\eqref{eq:EIk1} at the top of the next page.
\bigformulatop{18}{
\begin{align}~\label{eq:EIk1}
\mathbb{E}\{|\mathcal{I}_1(k,\ell)|^2\} &=
\eta_q\sum_{k'=0}^{k_{pq}}
\sum_{\ell'=0}^{\ell_{pq}}
b[k',\ell']
\sum_{c \notin \mathcal{K} }
\mathbb{E}\bigg\{\Big|x_{dq}[(k-k'+c)_N,(\ell-\ell')_M]\Big|^2\bigg\}
\mathbb{E}\bigg\{\Big|h_{pq}[(k-k')_N,(\ell-\ell')_M]\alpha[k,\ell,c]\Big|^2\bigg\}\nonumber\\
&=\Pdtq\eta_q \sum_{k'=0}^{k_{pq}}
\sum_{\ell'=0}^{\ell_{pq}}
b[k',\ell']\sum_{c \notin \mathcal{K} }
\mathbb{E}\bigg\{\Big|h_{pq}[(k-k')_N,(\ell-\ell')_M]\Big|^2\bigg\} \Big|\alpha[k,\ell,c]\Big|^2\nonumber\\
&=\Pdtq\eta_q\sum_{i=1}^{L_{pq}}
\mathbb{E}\Big\{\big|h_{pq,i}\big|^2\Big\}
\bigg|
\sum_{c \notin \mathcal{K} }\alpha[k,\ell,c]\bigg|^2.
\end{align}
}

We notice that in the delay domain, only $\ell_{pq}+1$ symbols before $\ell$ affect the received symbol in $\ell$. However, due to fractional Doppler shift, all data symbols outside the guard space $\mathcal{K}$ interfere with the received symbol on index $k$.\footnote{By using a full guard space~\cite{Raviteja:TVT:2019}, this interference term would vanish. However, it requires a higher signalling overhead of $(2\ell_{max}+1)N$ compared to that of the scheme in~\eqref{eq:datapatern}.} It can be checked that for $k\in[k_q-k_{max}-\hat{k},k_q+k_{max}+\hat{k}]$ and $c \notin \mathcal{K}$, $\big|\alpha[k,\ell,c]\big|$ lies in its sidelobe and becomes almost a constant. For
the case of a rectangular window, $\big|\alpha[k,\ell,c]\big|^2 \approx 1/N$~\cite{Kwan:TCOM:2021}, which results in
\setcounter{equation}{19}
\begin{align}~\label{eq:window2}
\sum_{c \notin \mathcal{K} }\big|\alpha[k,\ell,c]\big|^2 \approx \frac{(N-4k_{max}-4\hat{k}-1)}{N^2}.
\end{align}
Therefore, we get
\begin{align}~\label{eq:EIk1:final}
\mathbb{E}\Big\{\big|\mathcal{I}_1(k,\ell)\big|^2\Big\}
&\!\approx\!
\Pdtq\eta_q\frac{(N\!-\!4k_{max}\!-\!4\hat{k}\!-\!1)}{N^2}
\!\sum_{i=1}^{\Lmk}\!
\mathbb{E}\Big\{\!\big|\hmki\big|^2\!\Big\}\nonumber\\
&=\Pdtq\eta_q\frac{(N-4k_{max}-4\hat{k}-1)}{N^2}
\sum_{i=1}^{\Lmk}\beta_{pq,i}.
\end{align}

By using similar steps, we can obtain
\begin{align}~\label{eq:EIk2:final}
\mathbb{E}\Big\{\big|\mathcal{I}_2(k,\ell)\big|^2\Big\} &\approx
\frac{1}{N}
\sum_{q'\neq q}^{K_u}
\eta_{q^\prime}
{{P_{q'}^{\mathtt{dt}}}}
\sum_{i=1}^{\Lmk} \beta_{pq',i}.
\end{align}

To this end, by substituting~\eqref{eq:EIk1:final} and~\eqref{eq:EIk2:final} into~\eqref{eq:cpqk}, $c_{pq,i}$ is obtained as~\eqref{eq:MMSE1:Final} at the top of the next page,
\bigformulatop{22}{
\begin{align}~\label{eq:MMSE1:Final}
c_{pq,i} &\approx \frac{\sqrt{\roplq\eta_q}\beta_{pq,i}}{\roplq\eta_q\beta_{pq,i}\!
	\!+\!  \eta_q\left(\frac{1}{N}\sum_{q'=1}^{K_u}\frac{\eta_{q^\prime}}{\eta_q}\rodtqr\sum_{i=1}^{L_{pq'}}\beta_{pq',i}
	\! -\rodtq\frac{(4k_{max}+4\hat{k}+1)}{N^2}\sum_{i=1}^{\Lmk}\beta_{pq,i}\right)+1},
\end{align}
}
where $\roplq=P_p/\Sn$ is the normalized SNR of each pilot symbol.
Moreover, it can be readily checked that
\setcounter{equation}{23}
\begin{align}~\label{eq:var:MMSE:EP}
\gamma_{pq,i}^{\EP} \triangleq \mathbb{E}\big\{|\hat{h}_{pq,i}^{\EP}|^2\big\} =\sqrt{\roplq \eta_q} \beta_{pq,i}c_{pq,i}.
\end{align}

\subsection{Superimposed Pilot-based Channel Estimation}
Superimposed training has been used to estimate the double-selective channels in  TF domain for single-/multi-antenna channel estimation~\cite{Tugnait:CLET:2003,Ghogho:SLP:2003}. To mitigate the mutual interference between data and pilots, periodic superimposed pilots with zero-mean data symbols~\cite{Tugnait:CLET:2003} and data-dependent superimposed training scheme~\cite{Ghogho:SLP:2003} have been investigated in the literature. Recently, the authors in~\cite{Prasad:TWC:2021} applied SP-based channel estimation to estimate DD domain channel for OTFS systems. As indicated in~\cite{Prasad:TWC:2021}, SP-sided TF designs in~\cite{Tugnait:CLET:2003,Ghogho:SLP:2003} cannot be applied to  OTFS systems directly, as the channel gain in the DD domain rapidly varies across frames. Therefore, to expand the SP-based channel estimation to OTFS systems, channel estimation and data detection must be performed within a frame, and in presence of mutual interference between data and pilots in the DD domain~\cite{Prasad:TWC:2021}. Therefore, to reduce the bit error rate (BER) and improve the SE, new data detection designs, e.g. iterative algorithms must be developed~\cite{Dataided:WCL:2021,Prasad:TWC:2021}. Moreover, the transmit power need to be optimally allocated between the data and pilot symbols.

In this method, the data symbol $x_{dq}[k,\ell]$ is superimposed on to the pilot symbol $\psi_{q}[k,\ell]$ in the DD domain as $x_q[k,\ell] = x_{dq}[k,\ell]+\psi_{q}[k,\ell]$. By arranging $ x_{dq}[k,\ell]$ and $\psi_{q}[k,\ell]$  in a matrix form as $\qX_{dq} \in \mathbb{C}^{M\times N}$ and $\boldsymbol{\Phi}_{q}\in \mathbb{C}^{M\times N}$, the transmit matrix of the $q$th user can be expressed as
\begin{align}
\qX_{q} = \qX_{dq}+\boldsymbol{\Phi}_{q}.
\end{align}

Assume that the zero mean entries of $\qX_{dq}$ and $\boldsymbol{\Phi}_{q}$ are i.i.d. and satisfy $\mathbb{E}\{|x_{dq}[k,\ell]|^2\} = \Pdtq$ and $\mathbb{E}\{|\psi_{q}[k,\ell]|^2\} = \Pplq$. Let $\qx_q = \mathrm{vec} (\qX_q)= \qx_{dq} +\boldsymbol{\psi}_q\in\mathbb{C}^{MN\times1}$, with $\qx_{dq}\in\mathbb{C}^{MN\times1}$ and $\boldsymbol{\psi}_q\in\mathbb{C}^{MN\times1}$ being the uplink data and pilot vectors of the $q$th user in the DD domain, respectively. Then, the received signal at the $p$-th AP in~\eqref{eq:yAPm:Vect} can be recast as
\begin{align}
\qy_p = \boldsymbol{\Xi}_{\psi q}\qh_{pq} + \sum_{q^\prime\neq q}^{K_u}\boldsymbol{\Xi}_{\psi q^\prime}\qh_{pq^\prime}+\sum_{q=1}^{K_u}\boldsymbol{\Xi}_{dq}\qh_{pq}
+\qw_p,
\end{align}
where the concatenated matrices $\boldsymbol{\Xi}_{\psi q^\prime}\in\mathbb{C}^{MN\times L_{pq^\prime}}$ and $\boldsymbol{\Xi}_{dq}\in\mathbb{C}^{MN\times L_{pq}}$ correspond to the pilot and date vector $\boldsymbol{\psi}_q$ and $\qx_{dq}$ (these two matrices are referred to as equivalent pilot and data matrices) are given by~\cite{KWAN:TWC:2021}
\begin{subequations}
	\begin{align}
	\boldsymbol{\Xi}_{\psi q} &=\sqrt{\eta_q}\left[\qT_{pq}^{(1)}\boldsymbol{\psi}_q,\ldots,\qT_{pq}^{(\Lmk)}\boldsymbol{\psi}_q\right] ~\label{eq:Xipq}\\
	\boldsymbol{\Xi}_{dq} &=\sqrt{\eta_q}\left[\qT_{pq}^{(1)}\qx_{dq},\ldots,\qT_{pq}^{(\Lmk)}\qx_{dq}\right],
	\end{align}
\end{subequations}
where $\qh_{pq} = [h_{pq,1},\dots,h_{pq,\Lmk}]^T\in\mathbb{C}^{\Lmk\times 1}$. Let us denote $\tilde{\qw}_p=\sum_{q^\prime\neq q}^{K_u}\boldsymbol{\Xi}_{\psi q^\prime}\qh_{pq^\prime}+\sum_{q=1}^{K_u}\boldsymbol{\Xi}_{dq}\qh_{pq}
+\qw_p$. The MMSE estimate of the channel vector $\qh_{pq}$ can be obtained as
\begin{align}\label{eq:MMSE:SP}
\hat{\qh}_{pq} =
\left(\boldsymbol{\Xi}_{\psi q}^\dag\qC_{\tilde{\qw}_p}^{-1}\boldsymbol{\Xi}_{\psi q} + \qC_{\qh_{pq}}^{-1}\right)^{-1}
\boldsymbol{\Xi}_{\psi q}^\dag\qC_{\tilde{\qw}_p}^{-1}\qy_p,
\end{align}
where $\qC_{\tilde{\qw}_p}=\mathbb{E}\{ \tilde{\qw}_p \tilde{\qw}_p^\dag\}$  and $ \qC_{\qh_{pq}}=\diag\{\beta_{pq,1},\ldots,\beta_{pq,\Lmk}\}$ represent the covariance matrices of the  noise-plus-interference vector $\tilde{\qw}_p$, and the DD-domain CSI vector ${\qh}_{pq}$. Note that in our work, we assume  these covariance matrices are perfectly known. This assumption is reasonable since these matrices can be accurately estimated by following, for example, the  Bayesian learning (BL)-based framework proposed in \cite{Srivastava:TWC:2022}. The BL-based approach iteratively estimates the elements of the covariance matrix using the well-established expectation maximization procedure and is shown to be a very accurate estimation method \cite{Srivastava:TWC:2022}.

\begin{figure}[t]
	\centering
	\includegraphics[width=1\textwidth]{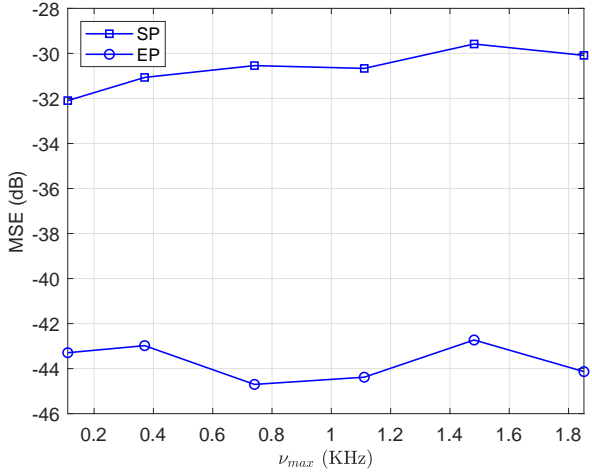}
	\caption{MSE comparison of the EP-aided and  SP-based channel estimation schemes as a function of $\nu_{\max}$ ($M_a=100, K_u=20, M=40, N=20$).}
	\label{fig. MSE}
\end{figure}

It is worth mentioning that the MMSE estimate in~\eqref{eq:MMSE:SP} requires to evaluate the inverse of a square matrix of size $\Lmk\times \Lmk$. Since the number of paths, i.e., $\Lmk$ is relatively small, the MMSE estimator does not incur much computational burden to the system.   We notice that OTFS enables the information symbols encoded in the DD domain (instead of the traditional TF domain) to experience a flat fading channel even when they are affected by multiple Doppler shifts present in high-mobility environments. By transforming back the received signals to the DD domain at the receiver, the channel estimation and data detection are performed in DD Domain. Thus, the impact of the time and frequency selectivity over the statistical expectations involved in~\eqref{eq:MMSE:SP} is alleviated.

In Figure~\ref{fig. MSE}, we compare the mean squared error (MSE) of the EP-aided and SP-based channel estimation schemes, defined as $\mathtt{MSE} = \mathbb{E}\{\|\qh_{pq}-\hat{\qh}_{pq}\|^2\}$, with increasing maximum Doppler $\nu_{\max}$. Large-scale fading coefficients are generated according to the 3GPP Urban Microcell model (please see Section~\ref{Sec:Numer} for more details). We observe that the MSE of both schemes remains almost constant for different values of $\nu_{\max}$, which is consistent with the results in~\cite{Raviteja:TWC:2018,Raviteja:TVT:2019,Prasad:TWC:2021}. This is because, similar to~\cite{Raviteja:TWC:2018,Raviteja:TVT:2019,Prasad:TWC:2021} we assume that the number of channel taps between each AP and user is constant with the Doppler spread and velocity, i.e., a higher Doppler spread does not imply an increase in the number of reflector clusters in DD domain. Hence, the proposed channel estimation methods, irrespective of the velocity, estimate the unknown channel gains and yield satisfactory performance.

\section{Performance Evaluation}~\label{Sec:Perf}
In this section, we derive new closed-form expressions for the downlink and uplink SEs, using the bounding technique from~\cite{Hien:cellfree}. We also quantify the power-scaling laws in the cases of imperfect CSI. Finally,  the achievable downlink SE under the assumption that the users have perfect CSI is derived as a benchmark. In order to derive the achievable downlink and uplink SEs, we assume that the proposed low complexity DD domain detector (LCD) in~\cite{Raviteja:TVT:2021} is utilized for detection of information symbols at each user and CPU, respectively. The LCD performs separate detection for each information symbol which yields lower complexity than the optimal detector, while achieving a SE close to that achieved with optimal joint demodulation~\cite{Raviteja:TVT:2021}.

\subsection{Preliminaries for Achievable SE Analysis}
We start by presenting some mathematical properties of the DD domain channel representation in~\eqref{eq:Hpq}, which will facilitate the subsequent achievable SE analysis.

\begin{lemma}\label{lemma:Tqi}
For the matrix $\qT_{pq}^{(i)}$, we have $\qT_{pq}^{(i)}\qT_{pq}^{(i)^{\dag}}=\qI_{\nm}$.
\end{lemma}
\begin{proof}
See Appendix~\ref{lemma:Tqi:proof}.
\end{proof}

\begin{lemma}\label{lemma:TqTqp}
For two different paths with different delay indices in~\eqref{eq:Hpq}, we have
	\begin{align}
	\Big[\qT_{pq}^{(i)} \qT_{pq'}^{(j)^\dag}\Big]_{(r,r)} =
0, \quad  (\ell_{pq,i}-\ell_{pq,j})_M\neq 0.
	\end{align}
\end{lemma}

\begin{proof}
See Appendix~\ref{lemma:TqTqp:proof}.
\end{proof}

\begin{lemma}\label{lemma:TqTqp:abs}
 For any two matrices of $\qT_{pq}^{(i)}$ and $\qT_{pq'}^{(j)}$ in~\eqref{eq:Hpq}, we have
	\begin{align}
\Bigg|\sum_{r'=1}^{MN} \bigg[\qT_{pq}^{(i)} \qT_{pq'}^{(j)^\dag}\bigg]_{(r,r')}\Bigg|^2=1.
	\end{align}
\end{lemma}
\begin{proof}
See Appendix~\ref{lemma:TqTqp:abs:proof}.
\end{proof}

Now, by utilizing Lemmas~\ref{lemma:Tqi}-\ref{lemma:TqTqp:abs}, the following two important results are presented for the superimposed pilot-based channel estimation method and transmit power constraint at the APs.

\begin{lemma}\label{Lemma:SP}
For the SP-based channel estimation in~\eqref{eq:MMSE:SP}, we have
\begin{subequations}
\begin{align}
&\mathbb{E}\{\tilde{\qw}_p\} = \boldsymbol{0}_{MN\times 1},\\
&\qC_{\tilde{\qw}_p}\!\!=\!
 \left(
 \!\sum_{q=1}^{K_u}{\eta_q}
 \!\Pdtq
\!\sum_{i=1}^{L_{pq}}\!
\beta_{pq,i} \!+\!
\!\sum_{q^\prime\neq q}^{K_u}\!{\eta_{q^\prime}}
\Pplqr
\!\sum_{i=1}^{L_{pq'}}\!
\beta_{pq',i} \!+\! \Sn\!\right)\!\qI_{MN}.~\label{eq:cov:SP}
\end{align}
\end{subequations}

Accordingly, the variance of the MMSE estimate of the channel vector $\qh_{pq}$ entries via the SP-based channel estimation can be expressed as~\eqref{eq:SP:MMSEche} at the top of the next page.
\bigformulatop{31}{
\begin{align}~\label{eq:SP:MMSEche}
\gamma_{pq,i}^{\SP} &\triangleq \mathbb{E}\big\{|\hat{h}_{pq,i}^{\SP}|^2\big\}
=
\frac{{\roplq}{\eta_{q}}\beta_{pq,i}^2}{ \roplq{\eta_{q}}\beta_{pq,i}\!+\!
	\sum_{q'\neq q}^{K_u}
	\!\sum_{i=1}^{\Lmkp}
	{\eta_{q^\prime}}
	\roplqr
	\beta_{pq',i}
	\!+\! \sum_{q'=1}^{K_u}
	\!\sum_{i=1}^{\Lmkp}
	{\eta_{q^\prime}}\rodtqr\beta_{pq',i} \!+\!1}.
\end{align}
}
\end{lemma}

\begin{proof}
	See Appendix~\ref{Lemma:SP:Proof}.
\end{proof}

\begin{lemma}\label{lemma:power}
	The power constraint in~\eqref{eq:AP:powconst}, is expressed as
	\setcounter{equation}{32}
	\begin{align}~\label{eq:dlpowercont}
	\sum_{q=1}^{K_u}\sum_{i=1}^{\Lmk}
	\eta_{pq}\gamma_{pq,i}^{\Pil}\leq 1,\hspace{3em}\Pil\in\{\EP,\SP\}
	\end{align}
\end{lemma}

\begin{proof}
	See Appendix~\ref{lemma:power:proof}.
\end{proof}
\subsection{Downlink SE Analysis}~\label{Sec:DLanalysis}
We assume that each user has knowledge of the channel statistics but not of the channel realizations~\cite{Hien:cellfree}. The signal received at the $q$th user~\eqref{eq:zqd} can be re-arranged to be suitable for detection of the $r$th entry of the received signal in DD domain with only statistical channel knowledge at users, according to~\eqref{eq:zqr} at the top of the next page,
\bigformulatop{33}{
\begin{align}~\label{eq:zqr}
z_{\dl,qr}
&=
\sqrt{\rho_d}
     \underbrace{
	        \mathbb{E}
	        \Big\{
				\sum_{p=1}^{M_a}
					\eta_{pq}^{1/2}[{\qH}_{pq}]_{(r,:)}
					[\hat{\qH}_{pq}^\dag]_{(:,r)}\!\Big\}}_{\text{desired signal}\triangleq\mathbb{DS}_{q,\dl}}s_{qr}
   \!+\!
\sqrt{\rho_d}
    \underbrace{
	           \Bigg(
                    \sum_{p=1}^{M_a}\!
	                  \etpq^{1/2}
	                   [{\qH}_{pq}]_{(r,:)}[\hat{\qH}_{pq}^\dag]_{(:,r)}
                        \!-\!\mathbb{E}\Big\{
                                            \sum_{p=1}^{M_a}\etpq^{1/2}
	                                       [{\qH}_{pq}]_{(r,:)\!}[\hat{\qH}_{pq}^\dag]_{(:,r)}
                                        \!\Big\}
              \Bigg)}_{\text{precoding gain uncertainty}\triangleq \mathbb{BU}_{q,\dl}}s_{qr}\nonumber\\
&\hspace{0em}
\! +\!
    \underbrace{
                \sqrt{\rho_d}
                \sum_{p=1}^{M_a}
	             \sum_{r'\neq r}^{MN}
                  \etpq^{1/2}
	               [{\qH}_{pq}]_{(r,:)} [\hat{\qH}_{pq}^\dag]_{(:,r')}s_{qr'}}_{\text{inter-symbol interference}\triangleq \mathbb{I}_{q1,\dl} }
+
\underbrace{
            \sqrt{\rho_d}
            \sum_{q'\neq q}^{K_u}
            \sum_{p=1}^{M_a}
	       \sum_{\substack{r'=1 }}^{MN}
            \eta_{pq'}^{1/2}
	       [{\qH}_{pq}]_{(r,:)} [\hat{\qH}_{pq'}^\dag]_{(:,r')}
            s_{q'r'}}_{\text{inter-user interference}\triangleq\mathbb{I}_{q2,\dl} }
 +
\underbrace{w_{\dl,qr}}_\text{AWGN},
\end{align}
}
where $s_{qr}=s_{q}[k,\ell]$, with $r=kM+\ell$ and $ k\in \mathbb{N}[0,N-1], \ell\in \mathbb{N}[0,M-1]$ denote the i.i.d. DD domain information symbols to be transmitted to the $q$th user, which satisfy $s_{qr}\sim \mathcal{CN}(0,1)$.

The sum of the second, third, forth and last term in~\eqref{eq:zqr} are treated as ``effective noise". Since $s_{qr}$ is independent of $\mathbb{DS}_q$ and $ \mathbb{BU}_q$, it can be readily checked  that the first and the second term of~\eqref{eq:zqr} are uncorrelated~\cite{Hien:cellfree}. Similarly it can be concluded that the third, fourth, and the noise terms of~\eqref{eq:zqr} are uncorrelated with the first term of~\eqref{eq:zqr}. Therefore, the effective noise and desired signal are uncorrelated. Hence, by using the fact that uncorrelated Gaussian noise represents the worst case, we obtain the following achievable downlink SE result:

\begin{theorem}~\label{Theor:DL}
An achievable downlink SE of the transmission from the APs to the $q$th user for any finite $M_a$ and $K_u$, is given by~\eqref{eq:Rdq:final} at the top of the next page,
\bigformulatop{34}{
\begin{align}~\label{eq:Rdq:final}
&R_{\dl,q}
=\frac{\taudl}{MN}
\sum_{r=1}^{MN}
  \log_2
     \Bigg(1+
      \frac
{\rho_{d}\left(\sum_{p=1}^{M_a}	
	\sum_{i=1}^{\Lmk} \etpq^{1/2}\gampqi \right)^2}
{  \rho_{d}
\Big(\!
\sum_{p=1}^{M_a}
\etpq
\sum_{i=1}^{\Lmk}
\beta_{pq,i}
\Big(\sum_{j=1}^{\Lmk}
\gampqj(\chpqij\! +\!\kapqij)
\!+\!\sum_{q'\neq q}^{K_u}
\sum_{j=1}^{\Lmkp}\!
\frac{\etpqr}{\etpq}\gampqrj
\Big)\Big)
\!+\!
1}\Bigg),
\end{align}
}
where $\taudl=\big({1 - \frac {N_\ul}{N_T}}\big)$,  $\chpqij=\big|[\qT_{pq}^{(i)}\qT_{pq}^{(j)^\dag}]_{(r,r)}\big|^2$, while $\kapqij=\Big|\sum_{r'\neq r}^{MN}\big[\qT_{pq}^{(i)}\qT_{pq}^{(j)^\dag}\big]_{(r,r')}\Big|^2$.
We note that $\taudl$ reflects the fact that, for OTFS frame of length $MN_T$ symbols, we spend $MN_\ul$ symbols for uplink transmission.
\end{theorem}

\begin{proof}
See Appendix~\ref{Theor:DL:Proof}.
\end{proof}
\begin{remark}
The achievable downlink expression in Theorem~\ref{Theor:DL}, depends on the delay and Doppler indices of the paths between the APs and the $q$th user through the $\chpqij$ and $\kapqij$. However, it is independent of the delay and Doppler indices between the APs and other interfering users (i.e., the inter-user interference  depends only on the power control coefficient and channel gain estimation). Moreover, the impact of the so called pilot contamination is captured by the channel estimation parameter $\gamma_{pq,i}^\Pil$.
\end{remark}

\begin{remark}
If the paths between the $p$th AP and $q$th user experience different delay indices, i.e., $\Lmki \neq \Lmkj$, by using Lemma~\ref{lemma:TqTqp} and Lemma~\ref{lemma:TqTqp:abs}, it can be readily checked that $\chpqij=0$ and $\kapqij=1$. Therefore, the precoding gain uncertainty in~\eqref{eq:vard:final} and the inter-symbol interference term in~\eqref{eq:Iq1} are  simplified to $\mathbb{E}\Big\{\big|\mathbb{BU}_{q,\dl}\big|^2\Big\}
= \sum_{p=1}^{M_a}
\sum_{i=1}^{\Lmk}
\betpqi\gampqi,$
and
$\mathbb{E}\Big\{\big|\mathbb{I}_{q1,\dl}\big|^2\Big\}
= \sum_{p=1}^{M_a}
\sum_{i=1}^{\Lmk}
\sum_{j\neq i}^{\Lmk}
\betpqi
\gampqj$, respectively. Accordingly, the achievable downlink SE in~\eqref{eq:Rdq:final} reduces to~\eqref{eq:Rdq:final:asymp} at the top of the next page.
\bigformulatop{35}{
\begin{align}~\label{eq:Rdq:final:asymp}
R_{\dl,q} &\!=\!
  \taudl\log_2
     \left(1\!+\!\frac{\rho_{d}\left(\sum_{p=1}^{M_a}	\sum_{i=1}^{\Lmk} \etpq^{1/2}\gampqi \right)^2}
     {\rho_{d}
	\!\sum_{p=1}^{M_a}
	\etpq
	\!\sum_{i=1}^{\Lmk}
	\betpqi
   \Big(\!
	\sum_{j=1}^{\Lmk}\!
	\gampqj \!+\!
    \sum_{q'\neq q}^{K_u}
	\sum_{j=1}^{\Lmkp}\!
	\frac{\etpqr}{\etpq}\gampqrj
    \!\Big)
	\!+\!
	1}\right).
\end{align}
}

Now, by inspecting~\eqref{eq:Rdq:final:asymp}, we see that with perfect knowledge of the delay and Doppler indices and using maximum-ratio beamforming, the impact of delay shift and Doppler spread can be efficiently mitigated in the special case of $\Lmki\neq\ell_{pq,j}$.
\end{remark}

In order to provide further insights into the system performance, we now assume that the channel coefficients between the $p$th AP and $q$th user are generated according to the distribution $\mathcal{CN}(0,1/L_{pq})$~\cite{Dataided:WCL:2021,Viterbo:WCL:2019,Mohammed:TVT:2021}. Accordingly, we have $\gamma_{pq}^{\Pil}=\gampqi$, $\forall i$. Moreover, assume that all APs transmit with full power, and at
the $p$th AP, the power control coefficients $\etpq$, $q=1,\ldots,K_u$, are the same, i.e.,
$\etpq = \big(\sum_{q'=1}^{K_u}\sum_{i=1}^{L_{pq}} \gampqri\big)^{-1} = \big(\sum_{q'=1}^{K_u} L_{pq'}\gamma_{pq'}^\Pil\big)^{-1}$, $\forall q=1,\ldots,K_u$. Then, the achievable downlink SE in~\eqref{eq:Rdq:final:asymp} is simplified to
\setcounter{equation}{36}
\begin{align}~\label{eq:Rdq:asym:2}
R_{\dl,q} &\!=\!
 \taudl \log_2
     \left(\!1\!+\frac{\rho_{d}\left(\sum_{p=1}^{M_a} \etpq^{1/2}L_{pq}\gamma_{pq}^\Pil \right)^2}
     {\rho_{d}\sum_{q=1}^{K_u}\sum_{p=1}^{M_a}
     	\etpq L_{pq}\gamma_{pq}^\Pil +1}\right).
\end{align}

To this end, if at any AP $p$, we have $L_{pq}=L_{q}$, i.e., the number of paths between the APs and $q$th user are the same,  then we have $\gamma_{q}^\Pil=\gamma_{pq}^\Pil$ and $\eta_{q}=\eta_{pq} $. Therefore, the approximation in~\eqref{eq:Rdq:asym:2} becomes
\begin{align}~\label{eq:Rdq:asym:3}
R_{\dl,q} &=
	\taudl\log_2 \left( 1 +
\frac{ \rho_{d}\eta_{p}\left( M_a L_{q}\gamma_{q}^\Pil \right)^2}
{  	K_u\rho_{d}M_a
     	\eta_{q}L_{q}\gamma_{q}^\Pil +1  }
\right).
\end{align}

\begin{corollary}~\label{Cor:DL}
If the transmit power of each AP in~\eqref{eq:Rdq:asym:3} is scaled down by a factor of $\frac{1}{M_a^2}$, i.e., $\rho_{d}=\frac{E_d}{M_a^2}$, where $E_d$ is fixed, we have
\begin{align}
R_{\dl,q}
\stackrel{M_a\rightarrow\infty}{\longrightarrow}&
\taudl
\log_2\left(1+E_d\eta_{q}L_{q}^2(\gamma_{q}^\Pil)^2\right)
\nonumber\\
&\hspace{-1.2em}
=\taudl\log_2\left(1+\frac{E_d}{K_u}L_{q}\gamma_{q}^\Pil\right).
\end{align}
\end{corollary}

Corollary~\ref{Cor:DL} implies that with imperfect CSI and a large $M_a$ and with the transmit power set to $\rho_{d}=\frac{E_d}{M_a^2}$, the achievable downlink SE of the  cell-free massive MIMO system with OTFS modulation
is equal to the performance of an interference-free single-input single-output (SISO) link with transmit power $\frac{E_d}{K_u}L_{q}\gamma_{q}$, without fast fading.

Finally, we present a lower bound of the achievable downlink SE  by using per-user-basis MMSE-SIC detector while treating co-user interference plus noise as uncorrelated Gaussian noise.  Assume that each user has access to the stochastic CSI $\bar{\qD}_{qq}=\mathbb{E}\{\qD_{qq}\}$ for data detection. Let $\qD_{qq'}=\sum_{p=1}^{M_a}\etpqr^{1/2} \qH_{pq}\hat{\qH}_{pq'}^{\dag}$. Therefore, by invoking~\eqref{eq:zqd} and applying MMSE-SIC detectors at the users,  the achievable downlink SE of the $q$th user can be obtained as
\begin{align}~\label{eq:DL:UPB}
{R}_{\dl,q}^{\mathtt{MMSE}}=
\frac{\taudl}{MN}
 { \log _{2} \det\left({  \qI_{MN} \!+\! \rho_d \bar{\qD}^{\dag}_{qq} \left ({\boldsymbol{\Psi }_{qq}}\right)^{-1} \bar{\qD}_{qq} }\right) },
\end{align}
where $\boldsymbol{\Psi}_{qq} =  \qI_{MN}+\rho_d \mathbb{E}\bigg\{\sum ^{K_u}_{q'\neq q}\qD_{qq'}\qD^{\dag}_{qq'}\bigg\}-\rho_d\bar{\qD}_{qq}\bar{\qD}^{\dag}_{qq}$.

With OTFS modulation, the computational complexity of the MMSE-SIC detector involves the inversion of the $MN\times MN$ effective channel gain matrix. This complexity is, however, $\mathcal{O}(M^4N^4)$, which is prohibitively high for practical implementation, especially when $M$ and/or $N$ is large. In case of LCD, since the information symbols are separately detected, the overall complexity of this detector is $\mathcal{O}(\nm\log(\nm))$~\cite{Raviteja:TVT:2021}.

\subsection{Uplink SE Analysis}
The CPU detects the received signal from the $q$th user, which is given by
\begin{align}\label{eq:ruq}
	\qr_{\ul,q}&=\sum_{p=1}^{M_a} \sqrt{\rodtq \eta_q}\hat{\qH}_{pq}^\dag \qH_{pq} \tilde{\qx}_q \nonumber\\
	&+
	\sum_{q'\neq q}^{K_u} \sum_{p=1}^{M_a} \sqrt{\rodtqr \eta_{q'}}\hat{\qH}_{pq}^\dag \qH_{pq'} \tilde{\qx}_{q'}+ \sum_{p=1}^{M_a}\hat{\qH}_{pq}^\dag \qw_p.
\end{align}

We assume that each information symbol of the $q$th user in the DD domain is detected separately~\cite{Raviteja:TVT:2021}. Accordingly, the desired signal $\tilde{\qx}_{qr}$ from $\qr_{\ul,q}$ in~\eqref{eq:ruq}, can be expressed as
\begin{align}~\label{eq:xqr:AP}
	\hat{\qx}_{qr} =& \sum_{p=1}^{M_a}\sqrt{\rodtq \eta_q} [\hat{\qH}_{pq}^\dag]_{(r,:)} [\qH_{pq}]_{(:,r)}\tilde{x}_{qr}\nonumber\\
	&
	+
	\sum_{p=1}^{M_a} \sum_{r'\neq r}^{MN} \sqrt{\rodtq \eta_q}[\hat{\qH}_{pq}^\dag]_{(r,:)} [\qH_{pq}]_{(:,r')}\tilde{x}_{qr'}\nonumber\\
	&+
	\sum_{q'\neq q}^{K_u} \sum_{p=1}^{M_a}\sum_{r=1}^{MN}
	\sqrt{\rodtqr \eta_{q'}}
	[\hat{\qH}_{pq}^\dag]_{(r,:)} [\qH_{pq'}]_{(:,r')}\tilde{x}_{q'r'}\nonumber\\
	&
	+\sum_{p=1}^{M_a}[\hat{\qH}_{pq}^\dag]_{(r,:)} \qw_p.
\end{align}

The CPU uses only statistical knowledge
of the channel when performing the detection. Therefore, the achievable uplink SE can be expressed as~\eqref{eq:Ruqdef} at the top of the next page,
\bigformulatop{42}{
\begin{align}\label{eq:Ruqdef}
 R_{\ul,q} \!=\!
\frac{\tauul}{MN}
\!\sum_{r=1}^{MN}\!
 \log_2 \left( 1\! +\! \frac{
	\left| \mathbb{DS}_{q,\ul}\right|^2}
{ 	\mathbb{E}\left\{\left|\mathbb{BU}_{q,\ul}\right|^2\right\} \!+\!\!
	\mathbb{E}\left\{\left|\mathbb{I}_{q1,\ul}\right|^2\right\}\! +\!\!
	\sum_{q'\neq q}^{K_u}\mathbb{E}\left\{\left|\mathbb{I}_{q2,\ul}\right|^2\right\} \!+\!\!
\mathbb{E}\left\{|w_{q,\ul}|^2\right\}}	\right)\!,
\end{align}
}
where $\tauul\in\big\{\tauulep, \tauulsp\big\}$ with $\tauulep= \Big(1-\frac{MN_\dl+\Nguard}{MN_T}\Big)$ corresponds to the SP-based channel estimation, while $\tauulsp = \Big(1-\frac{N_\dl}{N_T}\Big)$ stands for SP-based channel estimation. The term  $\tauul$ in~\eqref{eq:Ruqdef} reflects the fact that out of $MN_T$ symbol of each OTFS frame, we spend $MN_\dl$ symbol for downlink transmission and $\Nguard$ symbols for EP-aided channel estimation scheme. Moreover, by invoking~\eqref{eq:xqr:AP}, we have
\setcounter{equation}{43}
\begin{subequations}
	\begin{align}
\mathbb{DS}_{q,\ul}&=	
\mathbb{E}\left\{
			\sum_{p=1}^{M_a}
				\sqrt{\rodtq \eta_q}
				[\hat{\qH}_{pq}^\dag]_{(r,:)} [\qH_{pq}]_{(:,r)}\right\},\\
	\mathbb{BU}_{q,\ul}
	&=\sum_{p=1}^{M_a}
	\sqrt{\rodtq \eta_q}
	[\hat{\qH}_{pq}^\dag]_{(r,:)} [\qH_{pq}]_{(:,r)} -
\mathbb{DS}_{q,\ul},\\
	\mathbb{I}_{q1,\ul}
	&=
	   \sum_{p=1}^{M_a}
	      \sum_{r'\neq r}^{MN}
	          \sqrt{\rodtq \eta_q}
	          [\hat{\qH}_{pq}^\dag]_{(r,:)} [\qH_{pq}]_{(:,r')},\\
	 \mathbb{I}_{q2,\ul}
	 &=
     \sum_{p=1}^{M_a}\sum_{r=1}^{MN}
	\sqrt{\rodtqr \eta_{q'}}[\hat{\qH}_{pq}^\dag]_{(r,:)} [\qH_{pq'}]_{(:,r')}, \\
	w_{q,\ul}
	&=
	\sum_{p=1}^{M_a}[\hat{\qH}_{pq}^\dag]_{(r,:)} \qw_p.
	\end{align}
\end{subequations}

\begin{theorem}~\label{theor:UL}
	An achievable uplink SE for the $q$th user with matched filtering
	detection, for any $M_a$ and $K_u$, is given by~\eqref{eq:Ru} at the top of the next page.
\bigformulatop{44}{
	\begin{align}~\label{eq:Ru}
	&	R_{\ul,q}\!=\!
		\frac{\tauul}{MN}
		\!\sum_{r=1}^{MN}\!\log_2
		\Bigg(1\!+\!\frac{\rodtq\eta_{q}\left(\sum_{p=1}^{M_a}	\sum_{i=1}^{\Lmk} \gampqi \right)^2}
{\eta_{q}\sum_{p=1}^{M_a}\sum_{i=1}^{\Lmk}\betpqi\Big(\sum_{j=1}^{\Lmk}(\chpqij\!+\!\kapqij)\rodtq\gampqj
\!+\!\sum_{q'\neq q}^{K_u}\sum_{j=1}^{\Lmkp}\frac{\eta_{q'}}{\eta_{q}}\rodtqr\gampqrj\Big)
\!+\!\sum_{p=1}^{M_a}	\sum_{i=1}^{\Lmk} \gampqi}
			\Bigg).
	\end{align}
	}
\end{theorem}

\begin{proof}
	The proof is similar to that of Theorem~\ref{Theor:DL} and thus omitted for brevity.
\end{proof}

\begin{remark}
	 The achievable uplink SE in~\eqref{eq:Ru}, when different paths corresponding to the channel between the $q$th user and $p$th AP undergo different delays (i.e., $\Lmki\neq\Lmkj$), reduces to~\eqref{eq:Ru2} at the top of the next page.
\bigformulatop{45}{
	\begin{align}~\label{eq:Ru2}
		R_{\ul,q}\!=\!
		\tauul\log_2
		\left(1\!+\!
		\frac{\rodtq \eta_q\left(\sum_{p=1}^{M_a}	\sum_{i=1}^{\Lmk} \gampqi \right)^2}
		{\eta_q\sum_{p=1}^{M_a}
			     \sum_{q'=1}^{K_u}
                 \sum_{i=1}^{\Lmk}\!
			     \sum_{j=1}^{\Lmkp}\!
			     \frac{\eta_{q'}}{\eta_q}\rodtqr\betpqi\gampqrj
			\!+\!\sum_{p=1}^{M_a}	\sum_{i=1}^{\Lmk} \gampqi}
		\right)\!.
	\end{align}
	}
\end{remark}

\begin{corollary}~\label{Cor:UL}
With equal uplink power control and channel gain profile $\mathcal{CN}(0,1/L_{pq})$, if the transmit power of each user is scaled down by a factor of $\frac{1}{M_a}$, i.e., $\rodtq=\frac{E_u}{M_a}$, where $E_u$ is fixed, we have
\setcounter{equation}{46}
\begin{align}
R_{\ul,q}&=\!
\tauul\log_2
\left(1\!+\!
\frac{\rodtq\eta_q\left(M_aL_q\gamma_{q}^\Pil \right)^2}
{\rodtq\eta_qM_a
		K_uL_q\gamma_{q}^\Pil
    	\!+M_a L_q \gamma_{q}^\Pil}
\right)
\nonumber\\
&\hspace{-1em}\stackrel{M_a\rightarrow\infty}{\longrightarrow}
\tauul\log_2
\left(1\!+\!
E_u\eta_q L_q\gamma_{q}^\Pil
\right).
\end{align}

\end{corollary}

The above result implies that when $M_a$ grows large, the inter-user interference, inter-symbol interference, and noise disappear. The transmit power at each user can be made proportional to $\frac{1}{M_a}$, without performance degradation.

We now present the achievable uplink SE expressions with MMSE-SIC detector with the assistance of fully and partially centralized processing at the CPU. By invoking~\eqref{eq:yAPm:Vect}, we rewrite the received signal at the CPU as for data detection as
		\vspace{-0.2em}
		\begin{align}~\label{eq:Hhatq}
			\underbrace{[\qy_1^T, \ldots, \qy_{M_a}^T]^T}_{\qy\in\mathbb{C}^{M_aMN}} =&\sum_{q=1}^{K_u}
			\sqrt{\eta_q}
			\underbrace{ [
				{\qH}_{1q}^T,
				\ldots,
				{\qH}_{M_aq}^T]^T}_{\qH_q \in\mathbb{C}^{M_aMN \times MN}}
			\qx_q \nonumber\\
			&+ \underbrace{[{\qw}_1^T, \ldots, {\qw}_{M_a}^T]^T}_{{\qw}_p\in\mathbb{C}^{M_aMN}},
		\end{align}
 where $\qH_q$ is the DD domain channel of the $q$-th user at the CPU. For the sake of presentation, in the following we set $\eta_q=1$. The CPU uses all CSI to design  an arbitrary receive combining vector $\qV_q \in \mathbb{C}^ { M_aMN\times MN}$ for detection of $\qx_q$ as
		\vspace{-0.2em}
		\begin{align}~\label{eq:xhatq}
			\hat{\qx}_q\! = \!
			\qV_q^\dag \hat{\qH}_q \qx_q
			\!+\!
			\qV_q^\dag \tilde{\qH}_q \qx_q
			\!+\!
			\sum_{q'\neq q}^{K_u}\qV_{q}^\dag {\qH}_{q'} \qx_{q'}
			\!+\!
			\qV_{q}^\dag {\qw}_p.
		\end{align}
 This scheme is called ``\emph{Level 4: Fully Centralized Processing}" in the literature~\cite{Emil:TWC:2020}.  Known combining designs, such as MR combining with $\qV_q\!=\!\hat{\qH}_q $ and MMSE combining with $\qV_q\! = \! \left(\sum_{q'=1}^{K_u}\!\big(\!{\hat{\qH}}_{q'} {\hat{\qH}}_{q'}^\dag \!+\! \qC_{q'}\!\big) \!+ \!\Sn\qI_{M_aMN} \right)^{\!\!\!-1}\!\! {\hat{\qH}}_{q}$ can be used at the CPU, where $\qC_{q}\! =\! \diag(\qC_{1q},\ldots \qC_{M_aq})$ with $\qC_{pq} = \mathbb{E}\{{\tilde{\qH}}_{pq} {\tilde{\qH}}_{pq}^{\dag}\}$ and ${\tilde{\qH}}_{pq} = {{\qH}}_{pq}-{\hat{\qH}}_{pq}$. The uplink SE achieved by fully centralized processing  can be computed by the Monte-Carlo method. However, MMSE combining has high computational complexity since it requires first the computation of a  $M_a MN\times M_aMN$ matrix inverse and then a matrix multiplication, where $M_a MN$ is a  large number. Moreover, in the case of multi-antenna APs, this computational complexity is  increased dramatically, as we need to compute a  $LM_a MN\times LM_aMN$ matrix inverse, where $L$ is the number of antennas per AP.

As an alternative, Local MMSE (L-MMSE) combining matrix can be designed at each AP based on the local CSI at each AP. L-MMSE combining matrix will reduce the computational complexity and backhaul overhead required for implementing centralized MMSE combining matrix at the CPU. Let $\qV_{pq}\in\mathbb{C}^{MN\times MN}$ denotes the combining matrix designed by $p$-th  AP for $q$-th UE  by exploiting the local estimate of $\hat{\qH}_{pq}$. The local estimate of ${\qx}_q$ at the $p$-th AP can be expressed as
\vspace{-0.2em}
\begin{align}~\label{eq:xhatq:LMMSE}
	\hat{\qx}_{pq} \!=\!
	\qV_{pq}^\dag \hat{\qH}_{pq} \qx_q
	\!+\!
	\qV_{pq}^\dag \tilde{\qH}_{pq} \qx_q
	\!+\!
	\sum_{q'\neq q}^{K_u}\qV_{pq}^\dag {\qH}_{pq'} \qx_{q'}
	\!+\!
	\qV_{pq}^\dag {\qw}_p.
\end{align}
Therefore, the L-MMSE combining matrix is given by~\cite{Emil:TWC:2020}
\begin{align}~\label{eq:Vq:LMMSE}
	\qV_{pq} =
	\left(\!\sum_{q'=1}^{K_u}\!\big({\hat{\qH}}_{pq'} {\hat{\qH}}_{pq'}^\dag \!+ \!\qC_{pq'}\big) \!+\! \Sn\qI_{MN} \!\right)^{\!\!-1}
	{\hat{\qH}}_{pq}.
\end{align}
The local estimates $\hat{\qx}_{pq}$ are sent to the CPU, where they are weighted by the large-scale fading decoding (LSFD) matrices $\qA_{pq}\in\mathbb{C}^{MN\times MN}$ as $\hat{\qx}_q =
\sum_{p=1}^{M_a}\qA_{pq}^\dag  \hat{\qx}_{pq}$. This scheme is called ``\emph{Level 3: Local Processing \& LSFD}" in the literature~\cite{Emil:TWC:2020}.  Define $\qA_q  \triangleq [\qA_{1q}^T,\ldots,\qA_{M_aq}^T]^T\in\mathbb{C}^{M_aMN\times MN}$ and $\qD_{qq'}\triangleq[\qV_{1q}^\dag\hat{\qH}_{1q'}^{\mathtt{DD}},\ldots,\qV_{M_aq}^\dag \hat{\qH}_{M_aq'}^{\mathtt{DD}}]\in\mathbb{C}^{M_aMN\times MN}$. Then, the achievable uplink SE for user $q$ with MMSE-SIC detector in level 3 is obtained as
\begin{align}~\label{SE:UL:MMSE3}
	R_{\ul,q}^{\mathtt{MMSE,3}} = \frac{\omega_{\ul}^{\mathtt{che}}}{MN}\log_2 \det \left(  \qI_{MN}+\qD_{q,(3)}^{\dag} (\boldsymbol{\Psi}_{q,(3)})^{-1}\qD_{q,(3)} \right),
\end{align}
 where $\qD_{q,(3)} \triangleq \qA_q^\dag \mathbb{E}\{\qD_{qq}\} $ and $\boldsymbol{\Psi}_{q,(3)} \triangleq  \sum_{q'=1}^{K_u} \qA_q^\dag \mathbb{E}\{ \qD_{qq'} {\qD}_{qq'}^\dag\}   \qA_q - \qD_{qq} \qD_{qq'}^\dag + \Sn\qA_q^\dag\qS_q\qA_q$ with $\qS_q=\diag\left(\mathbb{E}\{\qV_{1q}^{\dag}\qV_{1q}\},\ldots,\mathbb{E}\{\qV_{M_aq}^{\dag}\qV_{M_aq}\}\right) \in\mathbb{C}^{M_aMN\times M_aMN}$. According to~\cite{Debbah:TWC:2022}, the achievable SE in~\eqref{SE:UL:MMSE3} is maximized when $\qA_q = \left(\sum_{q'=1}^{K_u} \mathbb{E}\{\qD_{qq'} \qD_{qq'}^\dag\} + \Sn\qS_q\right)^{-1}\mathbb{E}\{\qD_{qq}\}$. Nevertheless, the computation of  a  $M_a MN\times M_aMN$ matrix inverse is required for the level 3 processing. For simplicity, the CPU can alternatively weight the local estimates $\{\hat{\qx}_{pq}, p=1,\ldots,M_a\}$ by taking the average of them to obtain the final decoding symbol as $\hat{\qx}_q = \frac{1}{M_a}\sum_{p=1}^{M_a} \hat{\qx}_{pq}$, which is equivalent to set $\qA_{pq}=\frac{1}{M_a}\qI_{MN}$. This method is termed as ``\emph{Level 2: Local Processing \& Simple Centralized Decoding}" in~\cite{Emil:TWC:2020}. The achievable uplink SE for $q$-th user  in level 2 with MMSE-SIC detector is obtained as
\begin{align}~\label{SE:UL:MMSE}
	R_{\ul,q}^{\mathtt{MMSE,2}} = \frac{\omega_{\ul}^{\mathtt{che}}}{MN}\log_2  \det \left(  \qI_{MN}+\qD_{q,(2)}^{\dag} (\boldsymbol{\Psi}_{q,(2)})^{-1}\qD_{q,(2)} \right),
\end{align}
where $\qD_{q,(2)} \triangleq\sum_{p=1}^{M_a} \mathbb{E}\{\qV_{pq}^\dag \qH_{pq}\} $ and $\boldsymbol{\Psi}_{q,(2)} \triangleq  \sum_{q'=1}^{K_u} \sum_{p=1}^{M_a} \sum_{p'=1}^{M_a} \mathbb{E}\{ \qV_{pq}^\dag {\qH}_{pq'} {\qH}_{p'q'}^\dag\qV_{p'q}\}  - \qD_{q,(2)} \qD_{q,(2)}^\dag + \sum_{p=1}^{M_a}\mathbb{E}\{\qV_{pq}^\dag{\qw}_p^\dag{\qw}_p\qV_{pq}\}$.

In Section~\ref{Sec:Numer},  we numerically evaluate the achievable SE of the MMSE-SIC detector with L-MMSE combining and simple centralized processing using Monte-Carlo simulations. Deriving a closed-form expression for the achievable SE in~\eqref{SE:UL:MMSE} and quantifying the power-scaling laws for the L-MMSE combining can be considered as an interesting future line of research.
\subsection{OFDM Modulation}
In the simulation results, we will compare the performance of  OTFS with the OFDM counterpart. As indicated in~\cite{Raviteja:TWC:2018}, the Heisenberg transform in~\eqref{eq:Xtf} generalizes the inverse fast Fourier transform (IFFT) of OFDM and maps the information symbols from the frequency domain to time domain. Therefore, by replacing the Heisenberg transform module by IFFT, CP addition, serial-to-parallel and digital-to-analog conversion, an OFDM transmitter is implemented. On the other hand at the receive side, the Wigner transform module is substituted with analog-to-digital, parallel-to-serial, CP removal and FFT operation. For a fair comparison, we assume that each OFDM symbol duration is $T=\tcp+T_0$, where $\tcp$ and $T_0$ denote the CP and data symbol duration, respectively ($\tcp$ is typically chosen as $\tcp \geq  \tau_{max}$). Therefore, the OFDM frame duration is $NT$.

\emph{Uplink Payload Data Transmission:} With OFDM modulation, the continuous-time transmit signal from the $q$th user is given by
\begin{align}~\label{eq:st:OFDM}
\bar{s}_q(t)=\sum_{n=0}^{N-1}\sum_{m=0}^{M-1}
\bar{X}_q[n,m] \bar{g}_{tx}(t-nT)
e^{j2\pi m \Delta f(t-\tcp-nT)},
\end{align}
where $\bar{X}_q[n,m]$ denotes the information symbol at time instant $n$ and sub-carrier $m$, $\bar{g}_{tx}(t)$ is one for $t\in[0,T]$ and zero otherwise. After passing $\bar{s}_q(t)$ through the time-frequency selective channel
\begin{align}~\label{eq:hmk:TD}
h_{pq}(t,\tau) = \sum_{i=1}^{\Lmk}\hmki\delta(\tau-\tmki)e^{j2\pi\nmki t},
\end{align}
the received signal i.e., $\bar{r}(t) = \int h(t,\tau)\bar{s}(t-\tau)d\tau $, is sampled at a rate $f_s=M \Delta f = \frac{M}{T}$ and the CP is removed. Finally, after applying the DFT at the OFDM demodulator, the $n$th received OFDM symbol at the output of the $p$th AP can be written as
\begin{align}
\bar{\qy}_p(n) = \sum_{q=1}^{K_u} \sqrt{\rodtq \bar{\eta}_q}\qH_{pq}^\ofdm\bar{\qx}_q(n) + \bar{\qw}_p(n),
\end{align}
where $0\leq\bar{\eta}_q\leq1$ is the uplink power control coefficient for the $q$th user,  	 $\bar{\qx}_q(n)\in\mathbb{C}^{M\times 1}$ denotes the $n$th transmit signal from the $q$th user, $\bar{\qw}_p(n)\in\mathbb{C}^{M\times 1}$ is the AWGN at the AP, and $\qH_{pq}^\ofdm =\qF_M \qH_{pq}^{\TF} \qF_M^\dag\in\mathbb{C}^{M\times M}$ denotes the channel matrix over time-varying channel, where the TF domain channel matrix $\qH_{pq}^{\TF} =\sum_{i=1}^{\Lmk}
\hmki\qQ_{pq}^{(i)}$ with
\begin{align}
\big[\qQ_{pq}^{(i)}\big]_{(m,n)} =&
\delta\left(\left(m-n-\frac{\tau_{pq,i}M}{T}\right)_M\right)
e^{\frac{j2\pi}{M}(n-1)\nu_{pq,i}},\nonumber\\
&\hspace{5em}\quad m,n=1,\ldots,M,
\end{align}
where $\delta(\cdot)$ denotes the Dirac delta function.
In order to process the received signal, all APs need to acquire the CSI, which is obtained by having the user transmit uplink pilot symbols to the APs. In OFDM modulation-based systems, the channel estimation can be performed either by inserting pilot tones into all of the OFDM sub-carriers with a specific period $D_t$ over the time, termed as block-type (BT) pattern, or by inserting pilot tones to each OFDM symbol and over the adjacent sub-carries with frequency interval $D_f$ sub-carriers (termed as comb-type pattern)~\cite{Bahai:OFDM:2002}. Note that $D_t$ and $D_f$ must satisfy $D_t<\frac{1}{2 \nu_{max}{T}}$ and $D_f<\frac{1}{\Delta f\tau_{max}}$ constraints, respectively.

Having the MMSE estimate of the channel matrix $\qH_{pq}^\ofdm$, denoted by $\hat{\qH}_{pq}^{\ofdm}$, the CPU detects the received signal from $q$th user, which is given by
\begin{align}\label{eq:ruq:ofdm}
	\bar{\qr}_{\ul,q}(n)\!=&
	\sum_{q'=1}^{K_u} \sum_{p=1}^{M_a} \!\sqrt{ \rodtqr \bar{\eta}_{q'}}\hat{\qH}_{pq'}^{\ofdm^\dag} \qH_{pq'}^\ofdm \bar{\qx}_{q'}(n)\nonumber\\
	&+ \!\!\sum_{p=1}^{M_a}\hat{\qH}_{pq}^{\ofdm^\dag}\bar{\qw}_p(n).
\end{align}


\emph{Downlink Payload Data Transmission:} The APs use maximum-ratio beamforming to transmit signals to $K_u$ users. The signal transmitted   from the $p$th AP is
\begin{align}~\label{eq:xqd:ofdm}
\bar{\qx}_{\dl,p}(n)= \sqrt{\rho_d}
\sum_{q=1}^{K_u}
\bar{\eta}_{pq}^{1/2}\hat{\qH}_{pq}^{\ofdm^\dag} \bar{\qs}_q(n),
\end{align}
where $\bar{\qs}_q\in\mathbb{C}^{M\times 1}$ is the intended signal vector for the $q$th user; $\bar{\eta}_{pq}$, $p=1,\ldots,M_a$, $q=1,\ldots,K_u$ are the power control coefficients chosen to satisfy  the  following power constraint at each AP as $\mathbb{E}\left\{\|\bar{\qx}_{\dl,p}(n)\|^2\right\} \leq \rho_d$~\cite{Hien:cellfree}.

The received signal at the $q$th user in DD domain can be expressed as
\begin{align}~\label{eq:zqd:ofdm}
\bar{\qz}_{\dl,q}(n)
&=\sqrt{\rho_d}\sum_{q'=1}^{K_u}\sum_{p=1}^{M_a}
\bar{\eta}_{pq'}^{1/2} {\qH}_{pq}^{\ofdm} \hat{\qH}_{pq'}^{\ofdm^\dag} \bar{\qs}_{q'}(n)
+\bar{\qw}_{\dl,q}(n),
\end{align}
where $\bar{\qw}_{d,q}\in\mathbb{C}^{M\times 1}$ is the AWGN vector at the user $q$.

\begin{remark}
It can be verified that the results presented in Lemmas~\ref{lemma:Tqi}-\ref{lemma:TqTqp:abs}, hold for $\bar{\qQ}_{pq}^{(i)} =\qF_M \qQ_{pq}^{(i)} \qF_M^\dag\in\mathbb{C}^{M\times M}$. Therefore, by applying a similar method as in Theorem~\ref{Theor:DL} and~\ref{theor:UL}, the achievable downlink and uplink SE for the OFDM based system can be derived.
\end{remark}

In addition to the achievable SE analysis, BER is also an important performance metric. We leave the BER performance evaluation of the system for future investigation. As an initial suggestion in this study, the iterative channel and data detection algorithm proposed in~\cite{Prasad:TWC:2021} can be used in the OTFS cell-free massive MIMO system with SP-based channel estimation, by employing the proposed LCD in~\cite{Raviteja:TVT:2021} instead of the message passing data detection algorithm. More specifically, at each iteration the MMSE estimate of the channel vector $\qh_{pq}$, $\forall 1 \leq p\leq M_a$ and $1 \leq q\leq K_u$ is obtained using~\eqref{eq:MMSE:SP}. Then, the estimated channel $\hat{\qH}_{pq}$ and received data $\qy_q$ are fed into the LCD algorithm. After OTFS symbol detection, we are able to cancel out the interference induced by the data symbols based on the estimated channel and the obtained symbols. Accordingly,  MMSE channel estimates are updated for continuing the iterative process. The iterative process can be terminated when the stopping  criterion is reached.

\begin{table}
	\caption{System Parameters} 
	\centering 
	\begin{tabular}{||c c c c ||}
		\hline
		Parameter & Value &  &  \\ [0.5ex]
		\hline\hline
	Carrier frequency ($f_c$) & $4$ GHz  &  &  \\
		\hline
	Sub-carrier spacing ($\Delta f$)  & $15$ KHz & &  \\
		\hline
			Bandwidth ($B$) & $B=M\Delta f$ & & \\
		\hline
			OTFS frame duration ($T_c$) & $T_c=N/\Delta f$ &  &  \\
		\hline
			OFDM symbol duration ($T$) & $T = \tcp+T_0$  &  &  \\
		\hline
		OFDM CP duration ($\tcp$)  & $\tcp = \tau_{\max}$&  &  \\
			\hline
		OFDM pilot time interval ($D_t$) & $D_t = \frac{1}{2\nu_{\max}T}$&  &  \\
		\hline
	\end{tabular}
	\label{tab:FronthaulParameter}
\end{table}

\section{Numerical Results and Discussion}~\label{Sec:Numer}
We consider an OTFS system with operating carrier frequency is $f_c=4$ GHz and  the sub-carrier spacing is $\Delta f=15$ kHz. The maximum moving speed in the scenario is $300$ kmph, yielding a maximum Doppler index with $k_{max}=9$. We consider two 3GPP vehicular models, namely (i) extended vehicular A (EVA) with $L_{pq}=9$ and  $\tau_{max}=2.5~\mu$s, and (ii) extended vehicular B (EVB) with  $L_{pq}=6$ and  $\tau_{max}=10~\mu$s~\cite{channelmodel}.  For each tap of the channel realization between the $p$-th AP and $q$-th user, we randomly select the Doppler index according to the uniform distribution, such that we have $-k_{\max}\leq k_{p,q,i}\leq k_{\max}$, $\forall p, q, 1\leq i\leq \Lmk$~\cite{KWAN:TWC:2021}. The Doppler shift can be also generated using Jakes' formula, as $\nu_{pq,i} = \nu_{\max} \cos(\theta_{pq,i})$, where $\nu_{\max}$ is the maximum Doppler shift determined by the user speed while $\theta_{pq,i}$ is uniformly distributed over $[-\pi, \pi]$~\cite{Raviteja:TWC:2018}.

We assume that $M_a$ APs and $K_u$ users are uniformly distributed at random within a square of size $1 \times 1~\text{km}^2$ whose edges are wrapped around to avoid the boundary effects.
We model the large-scale fading coefficients according to the 3GPP Urban Microcell model. Therefore, we set   $\beta_{pq,i} = 10^{\frac{\mathrm{PL}_{pq,i}}{10}} 10 ^{\frac{F_{pq}}{10}}$, where $10^{\frac{\mathrm{PL}_{pq,i}}{10}}$ represents the path-loss and $10 ^{\frac{F_{pq}}{10}}$ denotes the shadowing effect with $F_{pq}\sim \mathcal{N}(0,4^2)$ (in dB)~\cite{Emil:TWC:2020}. Here $\mathrm{PL}_{pq,i}$ (in dB) is given by~\cite{Emil:TWC:2020}
\begin{align}
\mathrm{PL}_{pq,i} =
-30.5-36.7\log_{10}\left(\frac{d_{pq}}{1 \text{m}}\right),
\end{align}
where $d_{pa}$ is the distance between user $q$ and AP $p$. Moreover, the correlation among the shadowing terms from the AP $p$ to different users is expressed as~\cite{Emil:TWC:2020}
\begin{align}
\mathbb{E}\{F_{pq}F_{p'q'}\} =
\left\{ \begin{array}{ll}
4^2 2^{-\delta_{qq'}/9 \text{m}}      & p=p',  \\
0           & p\neq p',
\end{array} \right.
\end{align}
where $\delta_{qq'}$ is the distance between user $q$ and user $q'$.
We further set the noise figure $F = 9$ dB, and thus the noise power $\Sn=-108$ dBm ($\Sn=k_B T_0 (M\Delta f) F$ W, where $k_B=1.381 \times 10^{-23}$ Joules$/$K is the Boltzmann constant, while $T_0=290$K is the noise temperature). Let $\tilde{\rho}_d =1$ W and $\Pmax = 0.2$ W be the maximum transmit powers of the APs and users, respectively. Moreover, $P_{q}^{\mathtt{Pil}}= \alpha^{\mathtt{che}} P_{\max} $  and $P_{q}^{\mathtt{dt}}=(1-\alpha^{\mathtt{che}})P_{\max}$, denote the pilot and data symbol power at the $q$th user, respectively, where $0\leq\alpha^{\mathtt{che}}\leq 1$ is the power allocation factor. The normalized maximum transmit powers ${\rho}_d$, $\rodtq$, and $\roplq$  are respectively calculated by dividing $\tilde{\rho}_d$, $P_{q}^{\mathtt{dt}}$, and $P_{q}^{\mathtt{Pil}}$ by the noise power $\Sn$. We assume that all APs transmit with full power, and at the $p$th AP, the power control coefficients $\eta_{pq}$, $q=1,\ldots,K_u$, are the same, i.e., $\eta_{pq} = \big(\sum_{q'=1}^{K_u}\sum_{i=1}^{L_{pq}} \gamma_{pq',i}\big)^{-1}$, $\forall q=1,\ldots,K_u$, while in the uplink, all users transmit with full power, i.e., $\eta_q=1$, $\forall q=1,\ldots,K_u$.

\begin{figure}[t]
	\centering
	\includegraphics[width=1\textwidth]{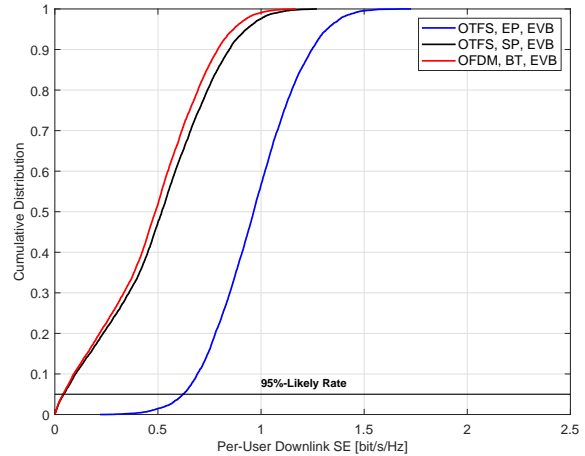}
	\caption{The cumulative distribution of the per-user downlink SE for the EVB and EVA model.}
	\label{fig:Fig1}
	\vspace{-0.5em}
\end{figure}
\begin{figure}[t]
	\centering
	\includegraphics[width=1\textwidth]{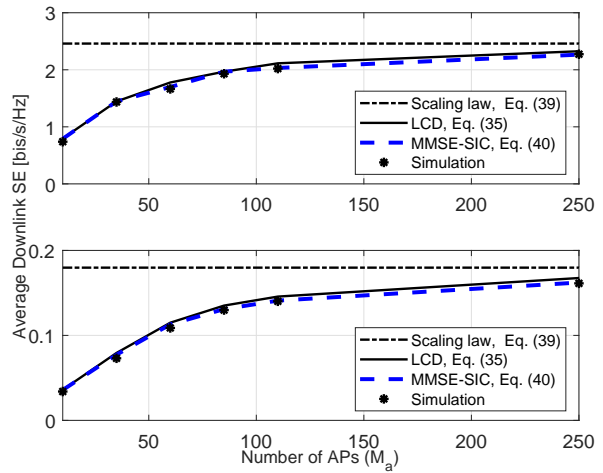}
	\caption{The average per-user downlink SE versus the number of APs and for EP (top) and SP (bottom) channel estimation.}
	\label{fig:Fig2}
	\vspace{0em}
\end{figure}

In Fig.~\ref{fig:Fig1}, we compare the per-user downlink SE of OTFS  against that of an OFDM modulation system for $M_a=100$ and $K_u=20$ for the EVA and EVB models. In this figure, in order to have a fair comparison between the channel estimation schemes, we set $\alpha^{\EP}= \alpha^{\SP}$ and, thus, the same power level is used for channel estimation at both EP and SP-based channel estimation schemes. It can be observed that OTFS modulation with EP-aided channel estimation method provides considerable gain over the OFDM modulation in terms of $95\%$-likely rate, while the gain of the SP-based channel estimation is negligible. Moreover, comparing the results for EVA and EVB model, we observe that both OFDM and OTFS modulation provide better performance over the EVA model. 	We notice that by decreasing $\tau_{max}$ the pilot overhead required for the EP-aided channel estimation in OTFS system and the length of CP for the OFDM modulation are decreased. However, since the OTFS system uses the uplink DD resources for  channel estimation, its downlink SE does not affected by decreasing $\tau_{max}$ (in this case the uplink SE of the system is improved as shown in Fig.~\ref{fig:Fig5} and Fig.~\ref{fig:Fig6}). Nevertheless, in the EVA model, the maximum number of path, $\Lmk$, is more than that of the EVB model and hence it is inventive that applying MRT precoding at the APs improves the performance of the OTFS-based systems compared to the EVB model.

\begin{figure}[t]
	\centering
	\includegraphics[width=1\textwidth]{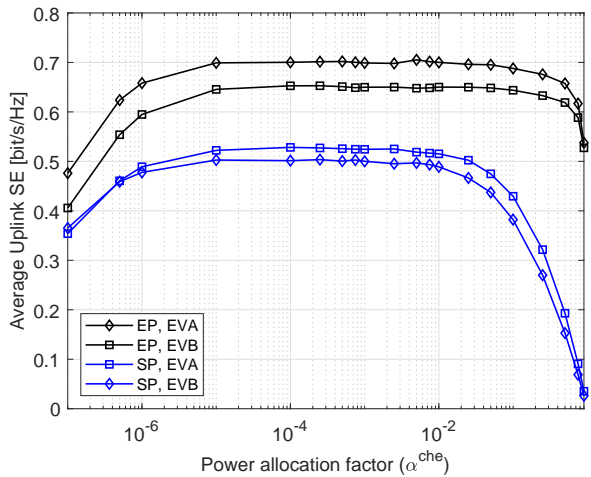}
	\caption{The average per-user uplink SE versus the power allocation coefficient.}
	\label{fig:Fig3}
	\vspace{0em}
\end{figure}

\begin{figure}[t]
	\centering
	\includegraphics[width=1\textwidth]{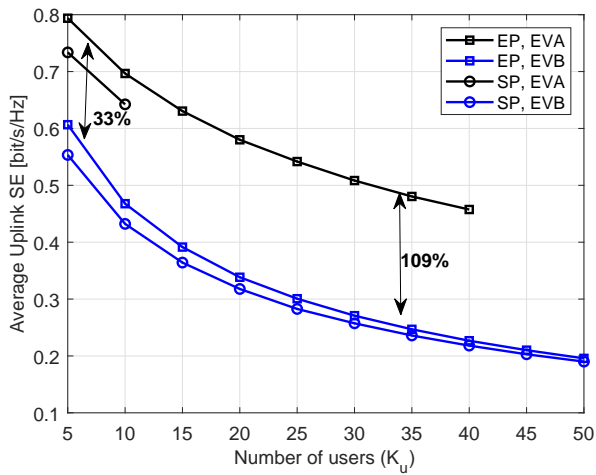}
	\caption{The average per-user uplink SE versus the number of users.}
	\label{fig:Fig4}
	\vspace{0em}
\end{figure}

In Fig.~\ref{fig:Fig2}, we compare the simulated average per-user downlink SE with the corresponding analytical expression in Theorem~\ref{Theor:DL}, when $\rho_d$ is scaled down with an $\frac{1}{M_a^2}$ factor. Results are presented for two different channel estimation schemes: EP and SP-based channel estimation. In this example, users transmit at full power. Moreover, the number of APs and users are assumed to be $M_a=100$ and $K_u=10$ respectively, with the OTFS frame consisting of $N=20$ symbols and $M=40$ sub-carriers. We further set $\beta_{pq,i}=1$, $\forall q=1,\ldots,K_u, p=1,\ldots,M_a, i=1,\ldots,\Lmk$. In line with Corollary~\ref{Cor:DL}, the downlink SEs saturate as $M_a$ tends to infinity in both scenarios. Figure~\ref{fig:Fig2} also illustrates the corresponding lower bound defined in~\eqref{eq:DL:UPB}, when stochastic CSI is available at the users.  It can be observed that the performance of the LCD is very close to its lower bound achieved by the MMSE-SIC detector.

Next, we evaluate the effect of power allocation between the data and pilot symbols on the uplink SE. Figure~\ref{fig:Fig3} shows the per-user uplink SE versus the power allocation coefficient $\alpha^{\mathtt{che}}$  for the EP and SP-based channel estimation schemes. In the simulations, we set $N=128$, $M=512$, $M_a=100$, and $K_u=10$. As expected, the performance of both channel estimation methods depends on tuning the data and pilot power (see~\eqref{eq:var:MMSE:EP} and~\eqref{eq:MMSE:SP}). The EP-aided channel estimation  outperforms the SP-based channel estimation, due to the guard intervals used around the impulse pilots, which facilitate the channel estimation process by decreasing the interference.  Another observation is that the EP-aided channel estimation provides higher uplink SE when the maximum delay spread of the channel decreases. This is intuitive, since, when $\tauul$ increases, the number of DD guard grids decreases and, thus, the variance of the MMSE channel estimate becomes large. Accordingly, the uplink SE improves, while in the case of SP-based channel estimation, the performance of the system with the EVB channel model outperforms that with the EVA model. This comes from the fact that, the variance of the MMSE estimates, as well as $\tauul$, are independent of the maximum delay spread, while the variance of the MMSE estimate becomes large when the maximum number of paths between the APs and users decreases.

\begin{figure}[t]
	\centering
	\includegraphics[width=1\textwidth]{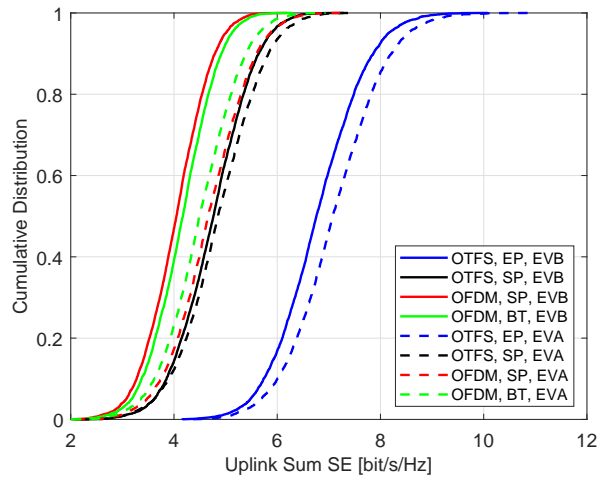}
	\caption{The cumulative distribution of the average uplink sum SE with LCD ($\Delta f=15$ kHz).}
	\label{fig:Fig5}
	\vspace{0em}
\end{figure}

\begin{figure}[t]
	\centering
	\includegraphics[width=1\textwidth]{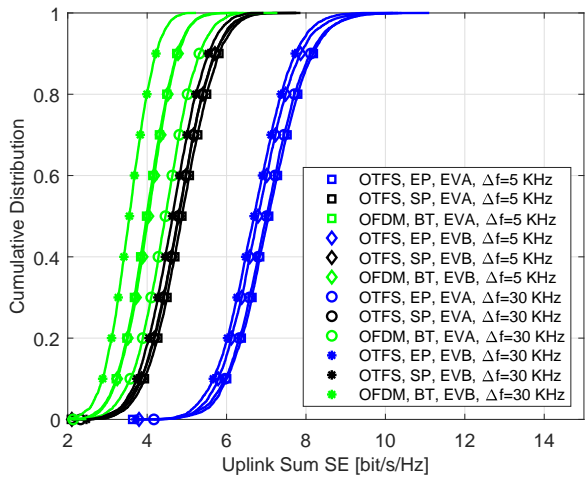}
	\caption{The cumulative distribution of the average uplink sum SE for different sub-carrier bandwidth.}
	\label{fig:Fig7_R3}
	\vspace{0em}
\end{figure}

Figure~\ref{fig:Fig4} shows the per-user uplink SE for both EP and SP-based channel estimation methods with EVA and EVB channel models. We consider the same setting as in Fig.~\ref{fig:Fig3} with $\alpha^{\mathtt{che}} = \Nguard/MN$. It can be observed, when the number of users increases, the SE reduction rate of the SP-based channel estimation is more significant compared to the EP-aided scheme. That is, by increasing $K_u$, the gap between the achievable SE via EP and SP-based scheme increases from $33\%$ to $109\%$, due to the increase in mutual interference between the data and pilot symbols. However, when $\tau_{max}$ increases, and with the proposed EP-aided channel estimation, only up to $K_u=10$ users can be served, while SP-based channel estimation scheme can support all users.  These results motivate the need to develop  user-centric clustering,  i.e., constructing serving clusters of cooperating APs separately for each user.

Figure~\ref{fig:Fig5} shows the cumulative distribution of the uplink sum SE for the EVA and EVB models, and for $N=128$, $M=512$, $M_a=100$, and $K_u=10$. Results for the OFDM-based system with BT and SP-based channel estimation have been included. When the channel maximum delay increases, the performance of the OTFS-based and OFDM-based systems is compromised in terms of $95\%$-likely rate. However, the gain of OTFS system over the OFDM modulation is increased. This is due to the fact that the CP overhead of the OFDM symbols is significantly increased as compared to the generic guard interval for the OTFS, which is introduced between two OTFS blocks to avoid inter-block interference.

\begin{figure}[t]
	\centering
	\includegraphics[width=1\textwidth]{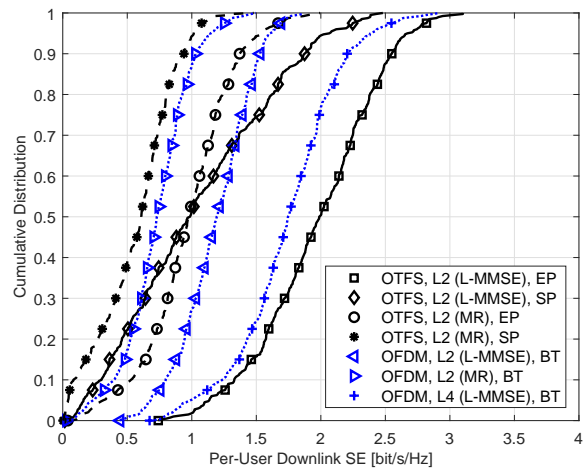}
	\vspace{-0.5em}
	\caption{The cumulative distribution of the  per-user uplink SE for MMSE-SIC detector.} 	\label{fig:Fig6}
	\vspace{-1em}
\end{figure}

We now study the impact of the sub-carrier bandwidth, $\Delta f$, when the number of sub-carriers $M$ is kept fixed. From Fig.~\ref{fig:Fig7_R3}, we observe that by increasing $\Delta f$ from $5$ kHz to $30$ KHz, the achievable SE of the OTFS is slightly degraded irrespective of the channel model. However, the SE achieved by the OFDM is improved considering the EVA model, while it is degraded considering the EVB model. This can be explained as follows: In the case of OTFS, the pilot overhead required for channel estimation is independent of $\Delta f$ and determined by the maximum delay and Doppler shifts. Therefore, by increasing $\Delta f$, only the total noise power is increased, which slightly degrades the system performance. In OFDM with BT pilot pattern, by increasing $\Delta f$, the time interval $D_t$ between the adjacent pilot symbols is increased ($D_t<\frac{\Delta f}{2 \nu_{max}}$). Hence, the pilot overhead required for the channel estimation process is decreased. However, the overhead required for the CP is increased. When $\tau_{max}$ is small (i.e., EVA model) the number of sub-carriers used for the CP is also small. Thus, by increasing $\Delta f$ the overall overhead required for the pilot and CP is decreased, and accordingly the achievable SE of the OFDM is improved. On the other hand, when $\tau_{max}$ is large (i.e., EVB model) the length of the  CP becomes large, especially when $\Delta f$ is high and $M$ is fixed. In this case, the increase in overhead required by the CP becomes dominant as compared to the decrease in pilot overhead, such that the overall overhead (pilot and CP) is increased. Therefore, by increasing $\Delta f$ the  SE achieved by OFDM is degraded.

In Fig.~\ref{fig:Fig6}, the  cumulative distribution function of the per-user uplink SE over different random realizations of the user locations and channels is plotted, when level-2 with L-MMSE and MR combining are deployed. In this figure,  the number of APs and users are assumed to be $M_a=100$ and $K_u=10$ respectively, with the OTFS frame consisting of $N=20$ symbols and $M=40$ sub-carriers. The L-MMSE combining with EP-aided and SP-based channel estimation, respectively achieves  2.6 and 4.3 fold higher 95\%-likely SE compared to the MR combining. We notice that level 4 with MMSE and level 3 with L-MMSE can provide higher SE with respect to the level 2. However, corresponding numerical results cannot be directly obtained for OTFS modulation, due to the high complexity of the large-scale matrix inversion involved in these scheme. Nevertheless, we included the per-user uplink SE for the OFDM modulation achieved by level-4 processing, as OFDM can afford the corresponding computational complexity. This result can be used as a benchmark for determining the performance gains achieved by the OTFS modulation. We observe that OTFS modulation with level-2 processing still outperforms OFDM modulation with level-4 processing.

\section{Conclusion}~\label{Sec:conclusion}
We analyzed the SE performance of cell-free massive MIMO  with OTFS modulation, taking into account the effects of channel estimation. Closed-form expressions for the average downlink and uplink SE were presented. Our results confirmed the superiority of OTFS against OFDM in improving the throughput of cell-free massive MIMO systems in high-mobility scenarios. Specifically, for the selected parameters and with EP-aided channel estimations, we achieved a 15 fold improvement in terms of $95\%$-likely per-user downlink SE  over the OFDM system. Moreover, SP-based channel estimation can be used in channels with large delay spread to support high number of users. Finally, we studied the impact of power allocation between the data and pilot symbols. Our results show that an optimized power allocation boosts the performance of the proposed channel estimation methods.

Extending these results to the more practical scenario and characterizing the impact of the delay and Doppler estimation errors on the achievable SEs, is an interesting future research direction. Moreover, in order to strike a balance between channel estimation quality and SE loss due to the impact of data symbols on the MMSE channel estimates,  the concept of max-min fairness optimization problem can also be applied to our framework.

\appendices
\section{Proof of Lemma~\ref{lemma:Tqi}}
\label{lemma:Tqi:proof}
By invoking~\eqref{eq:Hpq}, we can write
	\begin{align}
	\qT_{pq}^{(i)}\qT_{pq}^{(i)^{\dag}}
	&= (\textbf{F}_N \otimes \textbf{I}_M)
	\boldsymbol{\Pi}^{\Lmki} \boldsymbol{\Delta }^{\Kmki+\kpmki}
	(\textbf{F}_N^\dag \otimes \textbf{I}_M)\nonumber\\
	&\hspace{1em}\times
	(\textbf{F}_N \otimes \textbf{I}_M)
	(\boldsymbol{\Delta }^{\Kmki+\kpmki})^\dag(\boldsymbol{\Pi}^{\Lmki})^\dag (\textbf{F}_N^\dag \otimes \textbf{I}_M)\nonumber\\
	&\stackrel{(a)}{=}(\textbf{F}_N \otimes \textbf{I}_M)
	\boldsymbol{\Pi}^{\Lmki} \boldsymbol{\Delta }^{\Kmki+\kpmki}
	(\boldsymbol{\Delta }^{\Kmki+\kpmki})^\dag\nonumber\\
	&\hspace{1em}\times
	(\boldsymbol{\Pi}^{\Lmki})^\dag
	(\textbf{F}_N^\dag \otimes \textbf{I}_M)
	\nonumber\\
	&=(\textbf{F}_N \otimes \textbf{I}_M)
	\boldsymbol{\Pi}^{\Lmki} (\boldsymbol{\Pi}^{\Lmki})^\dag (\textbf{F}_N^\dag \otimes \textbf{I}_M)
	\nonumber\\
	&\stackrel{(b)}{=}(\textbf{F}_N \otimes \textbf{I}_M)
	  (\textbf{F}_N^\dag \otimes \textbf{I}_M)=\qI_{MN},
	\end{align}
where (a) follows from the identity $(\qA \otimes \qB)(\qC \otimes \qD) = (\qA\qC \otimes \qB\qD)$ and $\qF_N^{\dag}\qF_N=\qI_N$~\cite{Zhang:matrixbook} and (b) follows from the fact that $\boldsymbol{\Pi}^{\Lmki}(\boldsymbol{\Pi}^{\Lmki})^\dag=\qI_{MN}$.

\section{Proof of Lemma~\ref{lemma:TqTqp}}
\label{lemma:TqTqp:proof}
By invoking the definition of the effective channel matrix in~\eqref{eq:Hpq}, we get
\begin{align}\label{eq:TpTq}
	&\qT_{pq}^{(i)}\qT_{pq'}^{(j)^\dag}=
	 (\qF_N \otimes \qI_M ) \nonumber\\
	&\times
	\underbrace{\boldsymbol{\Pi}^{\Lmki}  \underbrace{\boldsymbol{\Delta }^{(\Kmki-\Kmkj)+(\kpmki-\kpmkj)}
    \left(\boldsymbol{\Pi}^{\Lmkj}\right)^{\dag}}_{\qB}}_{\qC}(\qF_N^\dag \otimes \qI_M ),
	\end{align}
where $\boldsymbol{\Pi}^{\Lmki}$
is a permutation matrix with $1$s in the $(u,(u-\Lmki)_{MN})$th entries, for $0\leq u \leq MN-1$, and zeros elsewhere. For the sake of notation simplicity, we express the diagonal matrix $\boldsymbol{\Delta }^{(\Kmki-\Kmkj)+(\kpmki-\kpmkj)}$ as
\begin{align}
&\boldsymbol{\Delta }^{(\Kmki-\Kmkj)+(\kpmki-\kpmkj)} =
\nonumber\\
&
\begin{bmatrix}
\boldsymbol{\Delta}^{(1)}_{( MN-\Lmkj)\times (MN-\Lmkj)}   &\boldsymbol{0}_{( MN-\Lmkj)\times \Lmkj}  \\
\boldsymbol{0}_{ \Lmkj \times ( MN-\Lmkj)}   &\boldsymbol{\Delta}^{(2)}_{\Lmkj\times \Lmkj} \\
\end{bmatrix},
\end{align}
where $\boldsymbol{\Delta}^{(1)} =\diag\{1,z_0,\ldots,z_0^{NM-\Lmkj-1}\}$ and $\boldsymbol{\Delta}^{(2)} =\diag\{z_0^{NM-\Lmkj},\ldots,z_0^{MN-1}\}$ with\\ $z_0 = z^{(\Kmki-\Kmkj)+(\kpmki-\kpmkj)}$. Let $\boldsymbol{\Pi}$ be some permutation matrix. Then, $\qA\boldsymbol{\Pi}$ is a matrix which has the columns of $\qA$ but in a permuted sequence~\cite{Zhang:matrixbook}. Therefore, $\qB$ in~\eqref{eq:TpTq} can be derived as
\vspace{-0.2em}
\begin{align}
\qB =
\begin{bmatrix}
\boldsymbol{0}_{( MN-\Lmkj)\times \Lmkj} &\boldsymbol{\Delta}^{(1)}_{( MN-\Lmkj)\times (MN-\Lmkj)}  \\
\boldsymbol{\Delta}^{(2)}_{\Lmkj\times \Lmkj} &\boldsymbol{0}_{ \Lmkj \times ( MN-\Lmkj)}
\end{bmatrix}.
\end{align}

Now, in order to derive $\qC$, we consider the case where $\Lmki>\Lmkj$  and represent $\qB$, as~\eqref{eq:Bmat} at the top of the next page,{\footnote{The case of $\Lmki<\Lmkj$ can be analyzed in a similar way, thus, it is omitted.}
\bigformulatop{66}{
\begin{align}\label{eq:Bmat}
\qB =
\begin{bmatrix}
\boldsymbol{0}_{(MN-\Lmki)\times \Lmkj}
&\boldsymbol{\Delta}^{(3)}_{(MN-\Lmki)\times(MN-\Lmki)}
&\boldsymbol{0}_{(MN-\Lmki)\times (\Lmki-\Lmkj)}
\\
\boldsymbol{0}_{(\Lmki-\Lmkj)\times \Lmkj}
&\boldsymbol{0}_{(\Lmki-\Lmkj)\times(MN-\Lmki)}
&\boldsymbol{\Delta}^{(4)}_{(\Lmki-\Lmkj)\times (\Lmki-\Lmkj)}
\\
\boldsymbol{\Delta}^{(2)}_{\Lmkj\times \Lmkj}
&\boldsymbol{0}_{\Lmkj\times(MN-\Lmki)}
&\boldsymbol{0}_{\Lmkj\times (\Lmki-\Lmkj)}
\end{bmatrix},
\end{align}
}
where $\boldsymbol{\Delta}^{(3)} =\diag\{1,z_0,\cdots,z_0^{MN-\Lmki-1}\}$ and $\boldsymbol{\Delta}^{(4)} =\diag\{z_0^{MN-\Lmki},\cdots,z_0^{MN-\Lmkj-1}\}$.
We note that, $\boldsymbol{\Pi}\qA$ is a matrix which has the rows of $\qA$ but in the permuted sequence~\cite{Zhang:matrixbook}. Therefore, $\qC$ in~\eqref{eq:TpTq}, for $\Lmki>\Lmkj$ can be derived as~\eqref{eq:Bmat2} at the top of the next page.
\bigformulatop{67}{
\begin{align}\label{eq:Bmat2}
\qC =
\begin{bmatrix}
\boldsymbol{0}_{(\Lmki-\Lmkj)\times \Lmkj}
&\boldsymbol{0}_{(\Lmki-\Lmkj)\times(MN-\Lmki)}
&\boldsymbol{\Delta}^{(4)}_{(\Lmki-\Lmkj)\times (\Lmki-\Lmkj)}
\\
\boldsymbol{\Delta}^{(2)}_{\Lmkj\times \Lmkj}
&\boldsymbol{0}_{\Lmkj\times(MN-\Lmki)}
&\boldsymbol{0}_{\Lmkj\times (\Lmki-\Lmkj)}
\\
\boldsymbol{0}_{(MN-\Lmki)\times \Lmkj}
&\boldsymbol{\Delta}^{(3)}_{(MN-\Lmki)\times(MN-\Lmki)}
&\boldsymbol{0}_{(MN-\Lmki)\times (\Lmki-\Lmkj)}
\end{bmatrix}.
\end{align}
}
Now, by inspecting $\qC$, it can be found there is only one non-zero entry in each column and row of $\qC$. Moreover, when $(\Lmki-\Lmkj)_M\neq 0$, all the entries on the main diagonal of $\qC$ are zero. Furthermore, the matrix $\qC$ possesses $(\Lmki-\Lmkj)-1$ all-zero subdiagonals and $(MN-\Lmki)$ all-zero superdiagonals.{\footnote{In the case of $\Lmkj>\Lmki$, $\qC$ consists of $(\Lmki-\Lmkj)-1$ and $(MN-\Lmki)$ all-zero superdiagonals and  subdiagonals, respectively.}}

We notice that $(\qF_N \otimes \qI_M )$ and $(\qF_N^\dag \otimes \qI_M )$ are squared matrices each of which has $N$ non-zero elements in $(u,(u)_M)$, $(u,(u+M)_M), \ldots, (u, (u+(N-1)M)_M )$ entries of its $u$th row. Therefore, the resulting matrix $\qT_{pq}^{(i)}\qT_{pq'}^{(j)^\dag}$ has non-zero elements in $(u, (u+\Lmkj-\Lmki)_M), \ldots, (u,(u+\Lmkj-\Lmki +(N-1)M)_M)$ entries of its $u$th row. Accordingly, since $\Lmki,\Lmkj <M$ (and thus $(\Lmkj-\Lmki)<M$)~\cite{Raviteja:TVT:2019}, it can be concluded that all entries on the diagonal of $\qT_{pq}^{(i)}\qT_{pq'}^{(j)^\dag}$ are zero.

When $\Lmki=\Lmkj$, the matrix $\qC$ can be readily obtained as
\vspace{-0.4em}
\setcounter{equation}{68}
\begin{align}
\qC =
\begin{bmatrix}
\boldsymbol{\Delta}^{(2)}_{\Lmkj\times \Lmkj} &\boldsymbol{0}_{ \Lmkj \times ( MN-\Lmkj)}\\
\boldsymbol{0}_{( MN-\Lmkj)\times \Lmkj} &\boldsymbol{\Delta}^{(1)}_{( MN-\Lmkj)\times (MN-\Lmkj)}
\end{bmatrix}.
\end{align}

Therefore, the matrix $\qT_{pq}^{(i)}\qT_{pq'}^{(j)^\dag}$ has non-zero elements in $(u, (u)_M), \ldots, (u,(u+(N-1)M)_M)$ entries of its $u$th row, which indicates that the main-diagonal of the matrix is non-zero. This completes the proof.

\vspace{-1em}
\section{Proof of Lemma~\ref{lemma:TqTqp:abs}}
\label{lemma:TqTqp:abs:proof}
Without loss of generality, we derive the first row of the matrix $\qT_{pq}^{(i)} \qT_{pq'}^{(j)^\dag}$ for $\ell_{pq,i}=\ell_{pq',j}.$\footnote{It is worth mentioning that, when $\ell_{pq,i}\neq\ell_{pq^\prime,j}$, by applying an appropriate circular shift to the final result, we get the representation of  the first matrix's row.} In this case, by invoking the result in the proof of Lemma~\ref{lemma:TqTqp}, we have
\vspace{-0.3em}
\begin{align}
&\qT_{pq}^{(i)} \qT_{pq'}^{(j)^\dag}=\nonumber\\
&
\underbrace{(\qF_N \otimes \qI_M )
\begin{bmatrix}
\boldsymbol{\Delta}^{(2)}_{\Lmkj\times \Lmkj} &\boldsymbol{0}_{ \Lmkj \times ( MN-\Lmkj)}\\
\boldsymbol{0}_{( MN-\Lmkj)\times \Lmkj} &\boldsymbol{\Delta}^{(1)}_{( MN-\Lmkj)\times (MN-\Lmkj)}
\end{bmatrix}}_{\qD}\nonumber\\
&\hspace{3em}\times
(\qF_N^\dag \otimes \qI_M ).
\end{align}
The first row of the $NM \times NM$ matrix $\qD$ can be readily obtained as
\vspace{-0.3em}
\begin{align}
&[\qD]_{(1,:)} = \nonumber\\
&\bigg[z_0^{MN-\Lmki},\underbrace{ 0, \ldots, 0}_{M-1}, z_0^{M-\Lmki}, \ldots, z_0^{M(N-1)-\Lmki}, \underbrace{ 0, \ldots, 0}_{M-1} \bigg].
\end{align}
Before proceeding, let us define
\vspace{-0.3em}
\begin{align}
\qa = z_0^{-\Lmki}\bigg[z_0^{MN}, z_0^{M}, \ldots, z_0^{M(N-1)} \bigg],
\end{align}
which contains the non-zero entries of $[\qD]_{(1,:)}$. Accordingly, the first row of matrix $\qT_{pq}^{(i)} \qT_{pq'}^{(j)^\dag}$ can be obtained as
\vspace{-0.3em}
\begin{align}~\label{eq:TqTqp1}
&\Big[\qT_{pq}^{(i)} \qT_{pq'}^{(j)^\dag}\Big]_{(1,:)} =\nonumber\\
&\bigg[\qa [\qF_N^\dag]_{(:,1)},\underbrace{ 0, \ldots, 0}_{M-1}, \qa [\qF_N^\dag]_{(:,2)}, \ldots, \qa [\qF_N^\dag]_{(:,N)}, \underbrace{ 0, \ldots, 0}_{M-1} \bigg].
\end{align}
To this end, by using~\eqref{eq:TqTqp1} and denoting the $k$th entry of $\qa$ by $a_k$, we get
\vspace{-0.2em}
\begin{align}~\label{eq:absTqTqp1}
&\Bigg|\sum_{r'=1}^{MN} \big[\qT_{pq}^{(i)} \qT_{pq'}^{(j)^\dag}\big]_{(1,r')}\Bigg|^2
=\frac{1}{N^2}
\Bigg|\sum_{j=0}^{N-1}\sum_{k=0}^{N-1} a_k e^{j2\pi\frac{jk}{N}}\Bigg|^2\nonumber\\
&=\frac{1}{N^2}
\sum_{k_1=0}^{N-1}
\sum_{k_2=0}^{N-1}
\sum_{j_1=0}^{N-1}
\sum_{j_2=0}^{N-1}
a_{k_1} a_{k_2}^*
e^{-j2\frac{\pi}{N}(j_1-j_2)(k_1-k_2)}\nonumber\\
&=\frac{1}{N^2}
\sum_{k_1=0}^{N-1}
a_{k_1}
\sum_{k_2=0}^{N-1}
a_{k_2}^*
N\delta(k_1-k_2)
N\delta(k_2-k_1) =1,
\end{align}
which completes the proof.

\vspace{-1em}
\section{Proof of Lemma~\ref{Lemma:SP}}
\label{Lemma:SP:Proof}
Since $\mathbb{E}\{\boldsymbol{\Xi}_{\psi q}\}=\mathbb{E}\{\boldsymbol{\Xi}_{dq}\}=\boldsymbol{0}_{MN\times \Lmk}$ and $\mathbb{E}\{\qh_{pq}\} =\boldsymbol{0}_{\Lmk\times 1}$, it can be readily checked that $\mathbb{E}\{ \tilde{\qw}_p\} =\boldsymbol{0}_{MN\times 1}$.

In order to derive the covariance matrix, since $\mathbb{E}\{ \tilde{\qw}_p\} =\boldsymbol{0}_{MN\times 1}$, we have
\vspace{-0.2em}
\begin{align}~\label{eq:covmat:SP}
\qC_{\tilde{\qw}_p}
&=\mathbb{E}\bigg\{
					\bigg(\sum_{q^\prime\neq q}^{K_u}
							\boldsymbol{\Xi}_{\psi q^\prime}\qh_{pq^\prime}
					 \bigg)
					 \bigg(\sum_{q^\prime\neq q}^{K_u}
					      \qh_{pq^\prime}^\dag\boldsymbol{\Xi}_{\psi q^\prime}^\dag
					 \bigg)
			\bigg\}
\\
&\hspace{1em}+
\mathbb{E}\bigg\{
				  \bigg(\sum_{q=1}^{K_u}\!
				      \boldsymbol{\Xi}_{dq}\qh_{pq}
				  \bigg)
				  \bigg(\sum_{q=1}^{K_u}\!
				     \qh_{pq}^\dag\boldsymbol{\Xi}_{dq}^\dag
			   	  \bigg)
			\bigg\}
\nonumber\\
&\hspace{1em}+
\!\mathbb{E}\bigg\{\!
				\bigg(\sum_{q^\prime\neq q}^{K_u}
						\boldsymbol{\Xi}_{\psi q^\prime}\qh_{pq^\prime}
				\!\bigg)
				 \bigg(\sum_{q=1}^{K_u}
				       \qh_{pq}^\dag\boldsymbol{\Xi}_{dq}^\dag
			 	 \!\bigg)
			\!\bigg\}
\nonumber\\
&\hspace{1em}+
\mathbb{E}\bigg\{\!
				\bigg(\!\sum_{q=1}^{K_u}\!
					\boldsymbol{\Xi}_{dq}\qh_{pq}
				\!\bigg)
				\bigg(\!\sum_{q^\prime\neq q}^{K_u}
					\qh_{pq^\prime}^\dag\boldsymbol{\Xi}_{\psi q^\prime}^\dag
				\!\bigg)
			\!\bigg\} \!+\! \Sn\qI_{MN}.\nonumber
\end{align}

Noticing the fact that data and pilot symbols are zero-mean and i.i.d. RVs, the third and fourth term in~\eqref{eq:covmat:SP} are zero. The first expectation term in~\eqref{eq:covmat:SP} can be expressed as
\begin{align}~\label{eq:covmat:SP:pilot}
&\mathbb{E}
		\bigg\{
			\bigg(
				\sum_{q^\prime\neq q}^{K_u}
				\boldsymbol{\Xi}_{\psi q^\prime}\qh_{pq^\prime}
			\bigg)
			\bigg(
				\sum_{q^\prime\neq q}^{K_u}\qh_{pq^\prime}^\dag
				\boldsymbol{\Xi}_{\psi q^\prime}^\dag
			\bigg)
		\bigg\}
\nonumber\\
&\hspace{2em}=\sum_{q^\prime\neq q}^{K_u}
  \sum_{q^{\prime\prime}\neq q}^{K_u}
			\underbrace{\mathbb{E}
				\bigg\{\boldsymbol{\Xi}_{\psi q^\prime}			
					\qh_{pq^\prime} \qh_{pq^{\prime\prime}}^\dag				
			         \boldsymbol{\Xi}_{\psi q^{\prime\prime}}^\dag
                \bigg\}}_{\boldsymbol{\Theta}_{q^\prime q^{\prime\prime}}}.
\end{align}

By invoking~\eqref{eq:Xipq} and assuming $q^\prime\!=\!q^{\prime\prime}$, we have $\Theta_{q^\prime q^{\prime}}\!=\!\Theta_{q^\prime q^{\prime\prime}}$, which can be obtained as
\begin{align}~\label{eq:covmat:SP:expect}
               \boldsymbol{\Theta}_{q^\prime q^{\prime}}
                &=\eta_{q'}\sum_{i=1}^{\Lmkp}
                \mathbb{E}\left\{h_{pq',i}h_{pq',i}^*\qT_{pq'}^{(i)}\boldsymbol{\psi}_{q'} \boldsymbol{\psi}_{q'}^{\dag}\qT_{pq'}^{(i)\dag}\right\}
                \nonumber\\
                &\hspace{2em}
                +\eta_{q'}\sum_{i=1}^{\Lmkp}
                 \sum_{j=1}^{\Lmkp}
                \mathbb{E}\left\{h_{pq',i}h_{pq',j}^*\qT_{pq'}^{(i)}\boldsymbol{\psi}_{q'} \boldsymbol{\psi}_{q'}^{\dag}\qT_{pq'}^{(j)\dag}\right\}\nonumber\\
                &\stackrel{(a)}{=}\sum_{i=1}^{\Lmkp}
\eta_{q'}\beta_{pq',i}\mathbb{E}\left\{\qT_{pq'}^{(i)}\boldsymbol{\psi}_{q'} \boldsymbol{\psi}_{q'}^{\dag}\qT_{pq'}^{(i)\dag}\right\}
\nonumber\\
                &=\eta_{q'}\sum_{i=1}^{\Lmkp}
\beta_{pq',i}\frac{\Trace(\boldsymbol{\psi}_{q'} \boldsymbol{\psi}_{q'}^{\dag})}{MN}\qT_{pq'}^{(i)}\qT_{pq'}^{(i)\dag}
\nonumber\\
&\hspace{0em}
\stackrel{(b)}{=}
\eta_{q'}\Pplqr\sum_{i=1}^{\Lmkp}
\beta_{pq',i}\qI_{MN},
\end{align}
where (a) holds as the channel entries are zero mean and i.i.d. RVs and (b) follows from  Lemma~\ref{lemma:Tqi}. Moreover,  by using the fact that channel entries of two different vectors are zero mean and i.i.d. RVs, we get $\mathbb{E}
				\big\{\boldsymbol{\Xi}_{\psi q^\prime}			
					\qh_{pq^\prime} \qh_{pq^{\prime\prime}}^\dag				
			         \boldsymbol{\Xi}_{\psi q^{\prime\prime}}^\dag
                \big\}=\boldsymbol{0}$ for $q^\prime\neq q^{\prime\prime}$.
Therefore, the expectation in~\eqref{eq:covmat:SP:pilot} can be derived as
\begin{align}
&\mathbb{E}
		\bigg\{
			\bigg(
				\sum_{q^\prime\neq q}^{K_u}
				\boldsymbol{\Xi}_{\psi q^\prime}\qh_{pq^\prime}
			\bigg)
			\bigg(
				\sum_{q^\prime\neq q}^{K_u}\qh_{pq^\prime}^\dag
				\boldsymbol{\Xi}_{\psi q^\prime}^\dag
			\bigg)
		\bigg\}
\nonumber\\
&\hspace{0em}=
\sum_{q^\prime\neq q}^{K_u}\eta_{q'}\Pplqr
\sum_{i=1}^{L_{pq'}}
\beta_{pq',i}
\qI_{MN}.
\end{align}
Following similar steps for deriving the first expectation term in~\eqref{eq:covmat:SP}, the second term in~\eqref{eq:covmat:SP} can be derived as
\begin{align}~\label{eq:covmat:SP:data}
\mathbb{E}
		\bigg\{
				\bigg(\sum_{q=1}^{K_u}
					\boldsymbol{\Xi}_{dq}\qh_{pq}
				\bigg)
				\bigg(\sum_{q=1}^{K_u}
					\qh_{pq}^\dag\boldsymbol{\Xi}_{dq}^\dag
				\bigg)
			\bigg\}
=
\sum_{q=1}^{K_u}\eta_{q}\Pdtq
\sum_{i=1}^{\Lmk}
\betpqi
\qI_{MN}.
\end{align}

To this end, by substituting~\eqref{eq:covmat:SP:pilot} and ~\eqref{eq:covmat:SP:data} into~\eqref{eq:covmat:SP}, the desired result in~\eqref{eq:cov:SP} is achieved.

\vspace{-0.6em}
\section{Proof of Lemma~\ref{lemma:power}}
\label{lemma:power:proof}
By invoking~\eqref{eq:xqd}, we have
\begin{align}~\label{eq:AP:powcons:final0}
\mathbb{E}\left\{\|\qx_{\dl,p}\|^2\right\}&=
\mathbb{E}
		\bigg\{
			\Big\|
				\sqrt{\rho_d}
				\sum_{q=1}^{K_u}
				\etpq^{1/2}\hHpqdg \qs_q
			\Big\|^2
		\bigg\}
\nonumber\\
&\hspace{-3em}=
\rho_d\mathbb{E}
			\bigg\{ \Trace
				\bigg(
					\Big[
						\sum_{q=1}^{K_u}\etpq^{1/2}\hHpqdg \qs_q
					\Big]
					\Big[
						\sum_{q'=1}^{K_u}\etpqr^{1/2}\qs_{q'}^\dag\hat{\qH}_{pq'}
					\Big]
				\bigg)
			\bigg\}
\nonumber\\
&\hspace{-3em}
=
\rho_d\sum_{q=1}^{K_u}
\etpq\Trace\left(\mathbb{E}\big\{\hHpqdg \hHpq\big\}\right),
\end{align}
where the equality holds since $s_{q,r}$ are assumed to be i.i.d. RVs. By using~\eqref{eq:Hpq}, and noticing that the channel gains of different paths are zero-mean i.i.d RVs, i.e., $\mathbb{E}\big\{\hmkihatc \hmkjhat\big\}=0$, $\forall i\neq j$, and then applying Lemma~\ref{lemma:Tqi},  we get $\mathbb{E}\left\{\|\qx_{\dl,q}\|^2\right\}= \rho_d\sum_{q=1}^{K_u}\sum_{i=1}^{\Lmk} \etpq\gampqi$, which completes the proof.

\vspace{-0.6em}
\section{Proof of Theorem~\ref{Theor:DL}}
\label{Theor:DL:Proof}
Since the channel model in~\eqref{eq:Hpq} consists of simply $MN$ parallel channels, an achievable SE per
channel input at the $q$th user is given by $R_{\dl,q} = \frac{\taudl}{MN}\sum_{r=1}^{MN} I_{\dl,qr} (\mathtt{SINR}_{\dl,qr}),$ with $I_{\dl,qr} (\mathtt{SINR}_{\dl,qr})= \log_2(1+\mathtt{SINR}_{\dl,qr})$, and $\mathtt{SINR}_{\dl,qr}$, i.e., the SINR at the $q$th user is
\begin{align}~\label{eq:SINRdq}
\mathtt{SINR}_{\dl,qr}\!=\!
\frac{
	\big|\mathbb{DS}_{q,\dl}
	\big|^2
}
{
	\mathbb{E}
	\big\{ \big|\mathbb{BU}_{q,\dl}\big|^2 \big\}
	\! \!+\!
	\mathbb{E}
	\big\{\big|\mathbb{I}_{q1,\dl}\big|^2\big\}
	\!+\!\!
	\mathbb{E}
	\big\{\big|\mathbb{I}_{q2,\dl}\big|^2\big\}
	\! +\!\!
	1/\rho_d    }.
\end{align}

We now proceed to derive the achievable downlink SE $R_{\dl,q}$. Noticing that $\Hpq=\sum_{i=1}^{\Lmk}\!
\hmki \Tpqi$ and $\hat{\qH}_{pq}=\sum_{i=1}^{\Lmk}\!\hmkihat \Tpqi$, we can derive $\mathbb{DS}_{q,\dl}$ as
\begin{align}~\label{eq:DS}
\mathbb{DS}_{q,\dl}
&=
\sum_{p=1}^{M_a}
\etpq^{1/2}
    \mathbb{E}
			\Big\{
				[\Hpq]_{(r,:)}[\hHpqdg]_{(:,r)}
			\Big\}
\nonumber\\
&
\stackrel{(a)}{=}\!
	\sum_{p=1}^{M_a}\!
			\etpq^{1/2}
						\mathbb{E}
							\Bigg\{
									\Big(
										\sum_{i=1}^{\Lmk}\!
										\hmkihat \big[\Tpqi\big]_{(r,:)}
									\Big)
									\nonumber\\
									&\hspace{2em}\times
									\Big(
										\sum_{j=1}^{\Lmk}\!
										\hmkjhatc \big[\Tpqjdg\big]_{(:,r)}
									\Big)
							\Bigg\}
\nonumber\\
&\hspace{0em}
\stackrel{(b)}{=}
			\sum_{p=1}^{M_a}
					\etpq^{1/2}
							\bigg(
								\sum_{i=1}^{\Lmk}
									\mathbb{E}
											\left\{
													|\hmkihat|^2
												\Big| \big[\Tpqi\big]_{(r,r)}\Big|^2
											\right\}
							\bigg)
\nonumber\\
&\hspace{0em}
\stackrel{(c)}{=}
\sum_{p=1}^{M_a}
		\etpq^{1/2}
				\Bigg(
					\sum_{i=1}^{\Lmk}
					\mathbb{E}\left\{|\hmkihat|^2 \right\}
				\Bigg)
=
\sum_{p=1}^{M_a}	
        \sum_{i=1}^{\Lmk}
            \etpq^{1/2}\gampqi,
\end{align}
where (a) follows by substituting $\epspqi=\hmki-\hmkihat$ and then using the fact $\epspqi$ and $\hmkihat$ are independent RVs and $\mathbb{E}\{\epspqi\}=0$; (b) follows from the fact that $\hmkihat$ are zero mean and independent; (c) follows from Lemma~\ref{lemma:Tqi}.

By using the fact that the variance of a sum of independent RVs is equal to the sum of the variances, $\mathbb{E}
	\big\{ \big|\mathbb{BU}_{q,\dl}\big|^2 \big\}$ in~\eqref{eq:SINRdq} can be  derived as
\begin{align}~\label{eq:vardd}
\mathbb{E}
\Big\{ \big|\mathbb{BU}_{q,\dl}\big|^2 \Big\} &=
\sum_{p=1}^{M_a}
\eta_{pq}
 \Bigg(
      \mathbb{E}
            \Big\{	
                \Big|
                    [\Hpq]_{(r,:)}[\hHpqdg]_{(:,r)}
                \Big|^2
            \Big\}
            \nonumber\\
            &\hspace{2em}
-         \Big|
            \mathbb{E}
               \Big\{
                   [\Hpq]_{(r,:)}[\hHpqdg]_{(:,r)}
               \Big\}
          \Big|^2
\Bigg).
\end{align}

Now, by using~\eqref{eq:Hpq}, and then applying Lemma~\ref{lemma:Tqi}, i.e., $\Big[\Tpqi\Tpqidg\Big]_{(r,r)}=1$,~\eqref{eq:vardd} can be expressed as
\begin{align}~\label{eq:vard}
&\mathbb{E}
\Big\{ \big|\mathbb{BU}_{q,\dl}\big|^2 \Big\} =
                \sum_{p=1}^{M_a}
                            \etpq
                                \bigg(
                                     \mathbb{V}_1-
                                                \Big(\sum_{i=1}^{\Lmk}
                                                    \gampqi
                                                \Big)^2
                                \bigg),
\end{align}
where
\begin{align*}
&\mathbb{V}_1=\nonumber\\
&
    \mathbb{E}
     \Bigg\{
            \bigg|\!
                \sum_{i=1}^{\Lmk}
                    \hmki\hmkihatc+
                \sum_{i=1}^{\Lmk}
                \sum_{j\neq i}^{\Lmk}
                    \big[\Tpqi\Tpqjdg\big]_{(r,r)}
                    \hmki\hmkjhatc
            \bigg|^2
     \Bigg\}.
\end{align*}
Before proceeding to derive $\mathbb{V}_1$, we define $X\triangleq\sum_{i=1}^{\Lmk}\hmki\hmkihatc$ and $Y\triangleq\sum_{i=1}^{\Lmk}\sum_{j\neq i}^{\Lmk}\hmki\hmkjhatc$. We note that $\mathbb{E}\left\{ |X+Y|^2\right\} = \mathbb{E}\left\{ |X|^2\right\} + \mathbb{E}\left\{ |Y|^2\right\}$ if $X$ and $Y$ are independent RVs and $\mathbb{E}\left\{Y\right\}=0$. It can be readily checked that $Y$ is a zero mean RV which is independent of $X$. Hence,  $\mathbb{V}_1$ can be evaluated as~\eqref{eq:EV10} at the top of the next page,
\bigformulatop{84}{
\begin{align}~\label{eq:EV10}
\mathbb{V}_1&=
\mathbb{E}
    \bigg\{
        \bigg|
            \sum_{i=1}^{\Lmk}
                \Big(
                    \epspqi\hmkihatc + \Big|\hmkihat\Big|^2
                \Big)
                    +\!
            \sum_{i=1}^{\Lmk}\sum_{j\neq i}^{\Lmk}
                \big[\Tpqi\Tpqjdg\big]_{(r,r)}
                    (\epspqi+\hmkihat)\hmkjhatc
        \bigg|^2
    \bigg\}
\nonumber\\
&=
\mathbb{E}
			\Bigg\{
				\Bigg|
					\sum_{i=1}^{\Lmk}
                        \Big(
                            \epspqi\hmkihatc +
				    	       \Big|\hmkihat\Big|^2
                        \Big)
				\Bigg|^2
			\Bigg\}
+\!
	\mathbb{E}
		\Bigg\{
			\Bigg|
				\sum_{i=1}^{\Lmk}\sum_{j\neq i}^{\Lmk}
					\big[\Tpqi\Tpqjdg\big]_{(r,r)}
					       (\epspqi+\hmkihat)\hmkjhatc
			\Bigg|^2
		\Bigg\}
\nonumber\\
&=
\sum_{i=1}^{\Lmk}
         	\left(
				\mathbb{E}
						\left\{
							\left|
							\epspqi\hmkihatc
							\right|^2
						\right\}
+
				\mathbb{E}
						\left\{
							\left|\hmkihat\right|^4
						\right\}
			\right)
+  \sum_{i=1}^{\Lmk}	
   \sum_{j\neq i}^{\Lmk}
               \mathbb{E}
 					\left\{
						\left|\hmkihat\right|^2
					\right\}
				\mathbb{E}
					\left\{
						\left|\hmkjhat\right|^2
					\right\}
\nonumber\\
&\hspace{1em}
+
\sum_{i=1}^{\Lmk}
\sum_{j\neq i}^{\Lmk}
			\mathbb{E}
					\left\{
						\mathfrak{R}[\Psi]
					\right\}
+
\sum_{i=1}^{\Lmk}
\sum_{j\neq i}^{\Lmk}
		\chpqij
			\mathbb{E}
					\left\{
						\left|\epspqi\right|^2
								+
						\left|\hmkihat\right|^2
					\right\}
			\mathbb{E}
					\left\{
						\left|\hmkjhat\right|^2
					\right\}, 						
\end{align}
}
where $\Psi\!=\!\bigg[\Big(\big[\Tpqi\Tpqjdg\big]_{(r,r)} \big[\Tpqj\Tpqidg\big]_{(r,r)}^*\Big)
\big(\hmkihat\hmkjhatc\big)^2 \bigg]$. To this end, by using the facts that $ \mathbb{E}\Big\{\big|\epspqi\big|^2\Big\} \!= \!\betpqi\!-\!\gampqi$, $\mathbb{E}\Big\{\big|\hmkihat\big|^4\Big\}\!=\!2\big(\gampqi\big)^2$, and $\mathbb{E}\big\{\mathfrak{Re}(\Psi)\big\}\! =\!0$,~\eqref{eq:EV10} is reduced to
\setcounter{equation}{85}
\begin{align}\label{eq:EV1}
\mathbb{V}_1=&
\sum_{i=1}^{\Lmk}
   \left(
        \gampqi (\betpqi-\gampqi)
        +
        2(\gampqi)^2
    \right)
    \nonumber\\
    &
+
\sum_{i=1}^{\Lmk} \sum_{j\neq i}^{\Lmk}
   \left(
         \gampqi\gampqj
         +
         \chpqij\betpqi\gampqj
   \right).
\end{align}
Therefore, by substituting~\eqref{eq:EV1} into~\eqref{eq:vardd} we obtain
\vspace{-0.2em}
\begin{align}~\label{eq:vard:final}
&\mathbb{E}
        \Big\{ \big|\mathbb{BU}_{q,\dl}\big|^2 \Big\} \!=\!
			\sum_{p=1}^{M_a}\eta_{pq}\sum_{i=1}^{\Lmk}
			\betpqi
			\Big(
				\gampqi
				\!+\!
				\sum_{j\neq i}^{\Lmk}
				\chpqij\gampqj
			\Big).
\end{align}

The inter-symbol interference term can be obtained as
\vspace{-0.1em}
\begin{align}
\mathbb{E}
\big\{|\mathbb{I}_{q1,\dl}|^2\big\}&=
\mathbb{E}
		\bigg\{
			\bigg|
				\sum_{p=1}^{M_a}
				\sum_{r'\neq r}^{MN}
				\etpq^{1/2}
				\big[\Hpq\big]_{(r,:)} \big[\hHpqdg\big]_{(:,r')}
			\bigg|^2
			\bigg\}
			\nonumber\\
&\hspace{-2em}=
\mathbb{E}
		\Bigg\{
			\Bigg|
				\sum_{p=1}^{M_a}
				\sum_{r'\neq r}^{MN}\etpq^{1/2}
				\Big(
				\sum_{i=1}^{\Lmk}\hmki \big[\Tpqi\big]_{(r,:)}\Big)
				\nonumber\\
				&
				\hspace{-1em}
				\times
				\Big(
				\sum_{j=1}^{\Lmk}\hmkjhatc \big[\Tpqjdg\big]_{(:,r')}\Big)
			\Bigg|^2
		\Bigg\}
\nonumber\\
&\hspace{-2em} =
\mathbb{E}
	\Bigg\{
		\Bigg|
			\sum_{p=1}^{M_a}
			\sum_{r'\neq r}^{MN}
			\etpq^{1/2}
						\bigg(
							\sum_{i=1}^{\Lmk}
							\hmki \hmkihat
                            \big[\Tpqi\Tpqidg\big]_{(r,r')}
\nonumber\\
&\hspace{-1em}
+
							\sum_{i=1}^{\Lmk}
							\sum_{ j\neq i}^{\Lmk}
							\hmki \hmkjhatc
							\big[\Tpqi\Tpqjdg\big]_{(r,r')}
						\bigg)
		\Bigg|^2
	\Bigg\}.
\end{align}
To this end, by using Lemma~\ref{lemma:Tqi}, i.e., $\big[\Tpqi\Tpqidg\big]_{(r,r')}=0$, and then by substituting $\hmki=\epspqi+ \hmkihat$,  we get
\vspace{-0.1em}
\begin{align}~\label{eq:Iq1}
\mathbb{E}
\big\{|\mathbb{I}_{q1,\dl}|^2\big\}&\!=\!
			\sum_{p=1}^{M_a}
				\sum_{i=1}^{\Lmk}
				\sum_{j\neq i}^{\Lmk}\!\!
               \etpq
				\kapqij
						\mathbb{E}
								\Big\{
									\Big|
										  \epspqi  \hmkjhatc
										\!+\! \hmkihat \hmkjhatc
									\Big|^2
								\Big\}\nonumber\\
&=
			\sum_{p=1}^{M_a}
			\sum_{i=1}^{\Lmk}
			\sum_{j\neq i}^{\Lmk}
                \etpq
				\kapqij
				\bigg(
					\mathbb{E}\Big\{\Big|\epspqi \Big|^2\Big\}
					\mathbb{E}\Big\{\Big|\hmkjhat\Big|^2\Big\}
					\nonumber\\
					&\hspace{1em}
					+
					\mathbb{E}\Big\{\Big|\hmkihat\Big|^2\Big\}
					\mathbb{E}\Big\{\Big|\hmkjhat\Big|^2\Big\}
				\bigg)
\nonumber\\
&=
\sum_{p=1}^{M_a}
    \sum_{i=1}^{\Lmk}
        \sum_{j\neq i}^{\Lmk}
                            \etpq
                            \kapqij
                            \betpqi
                            \gampqj,
\end{align}
where the first equality holds as the variance of a sum of independent RVs is equal to the sum of the variances and the second equality holds since the zero mean $\epspqi$ is independent of $\hmkihat$.
Noticing that the channel gains of different users are independent zero mean RVs, and applying Lemma.~\ref{lemma:TqTqp:abs}, the inter-user interference term can be derived as
\begin{align}~\label{eq:Iq2}
\mathbb{E}
\big\{|\mathbb{I}_{q2,\dl}|^2\big\}&=
\mathbb{E}
		\bigg\{
			\bigg|
				\sum_{p=1}^{M_a}
				    \sum_{q'\neq q}^{K_u}
				        \sum_{\substack{r'=1 }}^{MN}
                            \etpqr^{1/2}
				              [\Hpq]_{(r,:)} [\hat{\qH}_{pq'}^\dag]_{(:,r')}
			\bigg|^2
		\bigg\}
\nonumber\\
&=
\sum_{p=1}^{M_a}
        \sum_{q'\neq q}^{K_u}
            \sum_{i=1}^{\Lmk}
                \sum_{j=1}^{\Lmkp}
                    \etpqr\betpqi\gampqrj.
\end{align}

Finally, by substituting~\eqref{eq:DS},~\eqref{eq:vard:final},~\eqref{eq:Iq1}, and~\eqref{eq:Iq2} into~\eqref{eq:SINRdq}, after some algebraic manipulations the desired result in~\eqref{eq:Rdq:final} is obtained.
\vspace{-0.4em}
\bibliographystyle{IEEEtran}

\end{document}